\documentclass[pra,twocolumn,floatfix,superscriptaddress,longbibliography,notitlepage]{revtex4-1}

\usepackage{graphicx, color, graphpap}
\usepackage{enumitem}
\usepackage{amssymb}
\usepackage{amsthm}
\usepackage{multirow}
\usepackage[colorlinks=true,citecolor=blue,linkcolor=magenta]{hyperref}
\usepackage[T1]{fontenc}
\usepackage{bbm}
\usepackage{thmtools,thm-restate}
\usepackage{verbatim}
\usepackage{mathtools}
\usepackage{titlesec}
\usepackage{amsmath}
\usepackage{bbold}
\usepackage[title]{appendix}

\usepackage[caption = false]{subfig}

\long\def\ca#1\cb{} 

\newcommand{\braket}[2]{\langle #1 \hspace{1pt} | \hspace{1pt} #2 \rangle}

\newcommand{\ketbra}[2]{| \hspace{1pt} #1 \rangle \langle #2 \hspace{1pt} |}

\newcommand{\ket}[1]{|#1\rangle}               
\newcommand{\bra}[1]{\langle #1|}              
\newcommand{\dya}[1]{\ket{#1}\!\bra{#1}}
\newcommand{\dyad}[2]{\ket{#1}\!\bra{#2}}        
\newcommand{\ip}[2]{\langle #1|#2\rangle}      

\newcommand{\PC}{\mathcal{P}}

\newcommand{\Tr}{{\rm Tr}}

\renewcommand{\geq}{\geqslant}
\renewcommand{\leq}{\leqslant}
\newcommand{\mte}[2]{\langle#1|#2|#1\rangle }

\renewcommand{\vec}[1]{\boldsymbol{#1}}  

\newcommand{\ad}{^\dagger}

\newcommand*{\id}{\openone}

{}
{}

\newtheorem{lemma}{Lemma}
\newtheorem{proposition}{Proposition}

\begin{document}

\title{Covariance matrix preparation for quantum principal component analysis}

\author{Max Hunter Gordon}
\affiliation{Instituto de F\'isica Te\'orica, UAM/CSIC, Universidad Aut\'onoma de Madrid, Madrid, Spain}
\affiliation{Theoretical Division, Los Alamos National Laboratory, Los Alamos, NM 87545, USA}

\author{M. Cerezo}
\affiliation{Information Sciences, Los Alamos National Laboratory, Los Alamos, NM 87545, USA}

\author{Lukasz Cincio}
\affiliation{Theoretical Division, Los Alamos National Laboratory, Los Alamos, NM 87545, USA}
\affiliation{Quantum Science Center, Oak Ridge, TN 37931, USA}

\author{Patrick J. Coles}
\affiliation{Theoretical Division, Los Alamos National Laboratory, Los Alamos, NM 87545, USA}

\begin{abstract}
Principal component analysis (PCA) is a dimensionality reduction method in data analysis that involves diagonalizing the covariance matrix of the dataset. Recently, quantum algorithms have been formulated for PCA based on diagonalizing a density matrix. These algorithms assume that the covariance matrix can be encoded in a density matrix, but a concrete protocol for this encoding has been lacking. Our work aims to address this gap. Assuming amplitude encoding of the data, with the data given by the ensemble $\{p_i, \ket{\psi_i}\}$, then one can easily prepare the ensemble average density matrix $\overline{\rho} = \sum_i p_i \dya{\psi_i}$. We first show that $\overline{\rho}$ is precisely the covariance matrix whenever the dataset is centered. For quantum datasets, we exploit global phase symmetry to argue that there always exists a centered dataset consistent with $\overline{\rho}$, and hence $\overline{\rho}$ can always be interpreted as a covariance matrix. This provides a simple means for preparing the covariance matrix for arbitrary quantum datasets or centered classical datasets. For uncentered classical datasets, our method is so-called ``PCA without centering'', which we interpret as PCA on a symmetrized dataset. We argue that this closely corresponds to standard PCA, and we derive equations and inequalities that bound the deviation of the spectrum obtained with our method from that of standard PCA. We numerically illustrate our method for the MNIST handwritten digit dataset. We also argue that PCA on quantum datasets is natural and meaningful, and we numerically implement our method for molecular ground-state datasets.
\end{abstract}

\maketitle


\section{Introduction}

Interpreting and analyzing large datasets is a technologically important task. Principal component analysis (PCA) can reduce the dimensionality of large datasets to improve their interpretability while minimizing information loss~\cite{jolliffe_principal_2016}. PCA was invented by Karl Pearson in 1901 as an analog of the principal axis theorem in mechanics~\cite{pearson_liii_1901}. It can be thought of as fitting an ellipsoid to the data, where the length of each axis of the ellipsoid quantifies the variance of the data along that axis. PCA is widely used in bioinformatics, facial recognition, quantitative finance, and many other applications.

PCA is typically performed by diagonalizing the covariance matrix of the training data and retaining only the largest-eigenvalue eigenvectors. The covariance matrix quantifies correlations between different features in the data, with matrix elements given by:
\begin{equation}
    Q_{jk} = E[ (X_j - E(X_j)) (X_k - E(X_k))]
\end{equation}
for random variables $X_j$ and $X_k$, with $E$ denoting expectation value. One can see that $Q = [Q_{jk}]$ is a positive semi-definite matrix: $Q\geq 0$. For large datasets, $Q$ is often low rank, due to redundancy or correlation between features.

Quantum computers are naturally suited to solve linear algebra problems due to the underlying linear mathematics of quantum mechanics. Indeed, in 2014, Lloyd et al.~\cite{lloyd2014quantum} proposed a quantum algorithm for performing PCA called quantum principal component analysis (quantum PCA). Quantum PCA has the possibility of exponential speedup (over classical algorithms) when $Q$ is low rank. Lloyd et al.'s algorithm uses multiple copies of a density matrix $\rho$ in order to diagonalize $\rho$ and read off its spectrum. A key assumption in this algorithm is that the covariance matrix $Q$ can be encoded into a density matrix $\rho$. The plausbility of this assumption arises from the fact that both $Q$ and $\rho$ are positive semi-definite. Nevertheless, an explicit method for encoding $Q$ into $\rho$ has not been provided.

More recently, near-term approaches to quantum PCA have been developed~\cite{larose2019variational,cerezo2020variational,verdon2019quantum,ezzell2022quantum} in the framework of variational quantum algorithms~\cite{cerezo2020variationalreview}. For example, Variational Quantum State Diagonalization~\cite{larose2019variational} unitarily rotates $\rho$ towards a diagonal form, while estimating the distance to being diagonal using two copies of $\rho$. The Variational Quantum State Eigensolver~\cite{cerezo2020variationalfidelity} uses a non-degenerate Hamiltonian to extract the principal components of $\rho$ and only requires a single copy of $\rho$. Once again, while these variational methods are more near-term, they still do not address the issue of encoding $Q$ in~$\rho$.

We note that studies of quantum PCA are especially timely given that Huang et al.~\cite{huang2021quantum} recently established that quantum PCA can achieve exponential quantum advantage, at least for quantum data analysis. While classical data are still subject to dequantization arguments~\cite{tang2021quantum,cotler2021revisiting}, this does not preclude the possibility of modest quantum speedups for classical data~\cite{arrazola2019quantum}. (We elaborate on these points in the Discussion section.)

In this work, we make a simple but technologically important observation. We consider a dataset that has been amplitude encoded~\cite{grover2000synthesis,grover2002creating,plesch2011quantum,schuld2018supervised,sanders2019black,nakaji2021approximate,marin2021quantum,zoufal2019quantum}, such that the dataset is described by an ensemble $\{p_i, \ket{\psi_i}\}$ of normalized quantum states $\ket{\psi_i}$. We note that it is straightforward to use classical randomness in order to prepare the ensemble average density matrix $\overline{\rho} = \sum_i p_i \dya{\psi_i}$, as shown in Fig.~\ref{fig:preparation}. We ask the question: How is $\overline{\rho}$ related to the covariance matrix $Q$? For classical datasets (i.e., datasets stored on classical devices), one can show that we have a precise equality $\overline{\rho} = Q$ whenever the dataset is centered (i.e., the mean values for all features are zero).

For uncentered classical datasets $\overline{\rho} = Q + M$ where $M$ is the outer product of the mean vector with itself. This implies that  using $\overline{\rho}$ in place of the true covariance matrix corresponds to performing so-called ``PCA without centering''~\cite{cadima2009relationships}, and we provide an interpretation of this as performing PCA on a symmetrized dataset. We derive relations that upper bound the deviation of the eigenvalues and eigenvectors of $\overline{\rho}$ and $Q$. For example, we show that the eigenvalues of $\overline{\rho}$ are interlaced with those of $Q$, and any eigenvector of $\overline{\rho}$ is shared with $Q$ as long as it is either orthogonal or parallel to the mean vector. 

However, we go a step further and argue that the equality $\overline{\rho} = Q$ holds more generally for all quantum datasets (i.e., datasets prepared on quantum devices from some physical process), due to the unphysical nature of global phase in quantum mechanics. The (somewhat subtle) argument is that any dataset that will be prepared on a quantum device will necessarily lose its global phase information, and consequently one can always assume that such datasets are inherently centered. Therefore, the equality $\overline{\rho} = Q$ holds for all datasets that are prepared on quantum devices. As a consequence, for quantum datasets, we provide a simple means to prepare the covariance matrix as a density matrix, filling in the missing ingredient for quantum PCA algorithms.

\begin{figure}[t]
    \centering
    \includegraphics[width =\columnwidth]{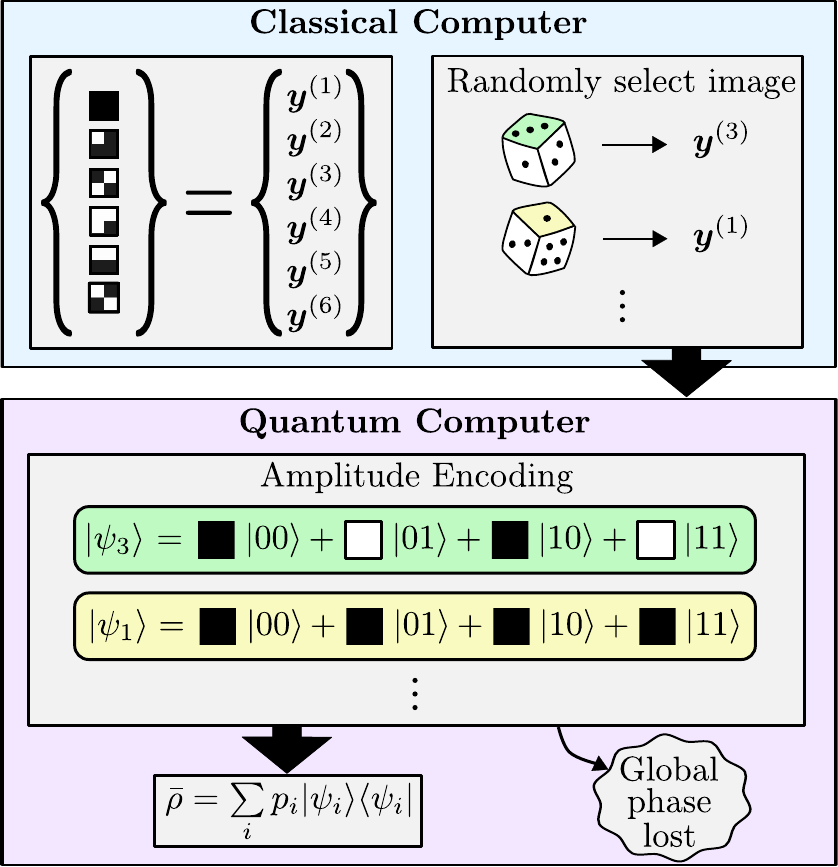}
    \caption{\textbf{Preparation of the ensemble average density matrix~$\overline{\rho}$}. For illustration, we show the well-known bars-and-stripes dataset. To prepare $\overline{\rho}$, images from this dataset are randomly selected according to a probability distribution $\{p_i\}$ and then amplitude encoded on a quantum device. The latter involves encoding each pixel's color into the amplitudes of a quantum state $\ket{\psi_i}$, expanded in the standard basis. Once the data are physically represented on a quantum device, the global phase information is lost and they can be represented as density matrices, of the form $\dya{\psi_i}$. The average state, averaged over many samplings from $\{p_i\}$, is what we call the ensemble average density matrix $\overline{\rho} = \sum_{i=1}^N p_i \dya{\psi_i}$. Preparing $\overline{\rho}$ simply adds a linear overhead (linear in the number of datapoints $N$) to the complexity of state preparation for each datapoint. See Section~\ref{sct_AnalysisOfSampling} for a detailed analysis of the complexity of our method when used as a subroutine of quantum PCA algorithms.}
    \label{fig:preparation}
\end{figure}

To illustrate our method, we numerically simulate both PCA and quantum PCA for the MNIST dataset of handwritten digits. In this case, loss of global phase information (due to amplitude encoding) amounts to losing the information about whether the image colors are white or black. Nevertheless, this loss of information is fairly trivial, and quantum PCA gives principal components that are very similar to those given by PCA. Moreover, quantum PCA performs as good as (or better than) PCA at compressing the MNIST dataset into a small number of features.

We also argue that it is natural to apply quantum PCA to quantum datasets, i.e., data that are inherently quantum mechanical and hence that do not require an amplitude encoding step. The covariance of complex random variables (such as quantum amplitudes) has a clear meaning, and therefore so does PCA on quantum states. Moreover, the loss of global phase information (mentioned above) does not apply to quantum datasets since there is no amplitude encoding step, so quantum PCA for quantum data is even more natural than for classical data. We illustrate quantum PCA for quantum datasets by applying it to a set of molecular ground states for different interatomic distances. Our quantum PCA simulation allows us to accurately compress these molecular ground states into a subspace of small dimension.

\section{Background}\label{sct_background}

Let us first give some background on covariance and PCA. We also discuss the covariance for complex random variables. This is particularly relevant to quantum systems, since quantum amplitudes are complex, in general.

\subsection{Covariance}

For two real random variables $X$ and $Y$, their covariance is given by
\begin{align}
    \text{Cov}(X,Y) &= E[ (X - E(X))(Y - E(Y))]\\
     &= E(XY) - E(X)E(Y)\,.
\end{align}
The sign of the covariance indicates whether the variables correlate (positive sign) or anti-correlate (negative sign), and the magnitude of the covariance quantifies the degree of correlation.

Now suppose that $X$ and $Y$ are complex random variables. In this case, it helpful to think of $X$ and $Y$ as vectors in the complex plane, with some associated randomness. Their covariance is given by:
\begin{align}
    \text{Cov}(X,Y) &= E[ (X - E(X)) \overline{(Y - E(Y))}]\\
     &= E(X\overline{Y}) - E(X)E(\overline{Y})\,,
\end{align}
where the overline indicates the complex conjugate. There is a geometric interpretation of this complex covariance, as follows. The covariance is a vector in the complex plane, with a direction and magnitude. The magnitude of this vector quantifies how correlated the two variables are, while the direction of this vector captures how out-of-phase the two variables are. When the covariance is postive, $X$ and $Y$ are in phase. When the covariance is negative, $X$ and $Y$ are completely out-of-phase. When the covariance is complex, $X$ and $Y$ are partially out-of-phase. For example, if the covariance is purely imaginary, then the two variables are 90 degrees out of phase. Figure~\ref{fig:Covariance} gives an illustration of this geometric interpretation for the covariance of complex random variables.

Hence, the covariance of complex random variables has a clear conceptual interpretation, and it generalizes the covariance for real random variables in perhaps the most natural way possible. This implies that applying the covariance to quantum amplitudes (which are complex numbers, in general) is conceptually meaningful. In turn, this implies that PCA on quantum states (i.e., vectors of quantum amplitudes) is conceptually meaningful.

If one has a set of random variables, then one can look at all of the pairwise covariances, and this is essentially the covariance matrix. Suppose that $\vec{X}= \{X_1, ..., X_d\}$ is a $d$-dimensional vector of complex random variables. Then the $(j,k)$-th entry of the covariance matrix $Q$ is given by:
\begin{align}
    Q_{jk} &= E[ (X_j - E(X_k)) \overline{(X_k - E(X_k))}]\\
     &= E(X_j \overline{X_k}) - E(X_k)E(\overline{X_k})\,.
\end{align}

\begin{figure}[t]
    \centering
    \includegraphics[width =\columnwidth]{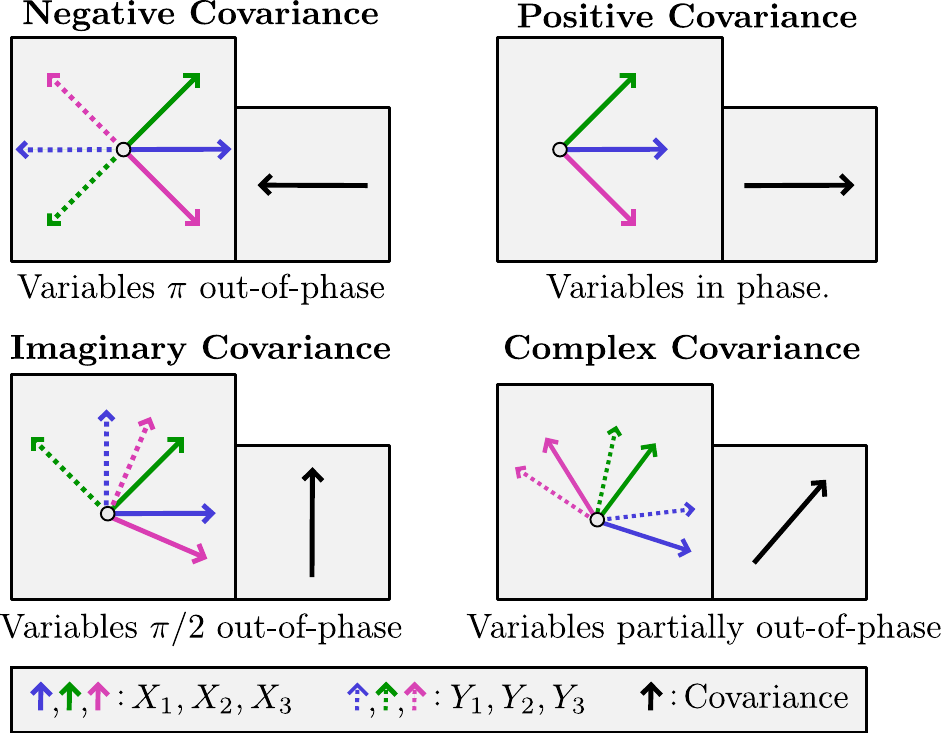}
    \caption{\textbf{Covariance for complex random variables.} For two complex random variables $X$ and $Y$, one can sample from their joint probability distribution $P(X,Y)$, and here we show three different samplings (with the purple, pink, and green colors). Each sampling leads to a pair of vectors in the complex plane (solid vector = $X$, dashed vector = $Y$). In the case of negative covariance, the vectors $X$ and $Y$ are $180$ degrees out of phase. For positive covariance, the vectors are in phase. Imaginary covariance corresponds to vectors that are $90$ degrees out of phase. In general, complex covariance arises from variables that are partially out of phase. In all cases, the direction of the covariance vector is indicated with a black arrow.}
    \label{fig:Covariance}
\end{figure}

\subsection{Principal component analysis}

In the context of PCA, one starts with a dataset $\{\vec{y^{(i)}} \}_{i=1}^N$ of $N$ data points. Each data point $\vec{y^{(i)}} $ is vector over a $d$-dimensional feature space:
\begin{align}
\vec{y^{(i)}} = \{y^{(i)}_1, ..., y^{(i)}_d\}^T\,.
\end{align}
In classical data analysis, the data typically consists of real numbers such that $y^{(i)}_j \in \mathbb{R}$. For quantum data, the $\vec{y^{(i)}}$ vectors would be quantum states and the $y^{(i)}_j$ would be complex amplitudes. Hence, for generality we can allow for $y^{(i)}_j \in \mathbb{C}$.

We can compute the empirical mean of the data as the following vector:
\begin{align}
\vec{\mu} = \{\mu_1, ..., \mu_d\}^T,
\end{align}
where $\mu_j = \frac{1}{N} \sum_{i=1}^N y^{(i)}_j$ is the empirical mean value for the $j$-th feature when the distribution over data points is uniform. More generally, it is possible that each data point $\vec{y^{(i)}} $ could come with an associated probability $p^{(i)}$, and the probability distribution $\{p^{(i)}\}$ might not be uniform. In that case we have
\begin{align}
\mu_j =  \sum_{i=1}^N p^{(i)} y^{(i)}_j\,.
\end{align}

For notational convenience, we define the centered dataset $\{\vec{b^{(i)}}\}_{i=1}^N$ by subtracting off the mean value to each data point:
\begin{align}
\vec{b^{(i)}} = \vec{y^{(i)}} - \vec{\mu}\,.
\end{align}
The covariance between the $j$-th and $k$-th feature is then given by:
\begin{align}
    Q_{jk} &= \sum_{i=1}^{N} p^{(i)} b^{(i)}_j \overline{b^{(i)}_k} \\
    &= \sum_{i=1}^{N} p^{(i)} (y^{(i)}_j - \mu_j) \overline{(y^{(i)}_k - \mu_k)}\\
    &= \left(\sum_{i=1}^{N} p^{(i)} y^{(i)}_j  \overline{y^{(i)}_k} \right)  - \mu_j \overline{\mu_k}\,.
    \label{eqnKjk}
\end{align}
In the special case of the uniform distribution, this is:
\begin{align}
    Q_{jk} &= \frac{1}{N}\sum_{i=1}^{N}  b^{(i)}_j \overline{b^{(i)}_k} \,.
\end{align}
(In some literature, one uses the factor $\frac{1}{N-1}$ instead of $\frac{1}{N}$ due to Bessel's correction. This normalization factor does not play a significant role in our analysis, so we will ignore this issue.)

The PCA method involves diagonalizing $Q$ and keeping the largest eigenavalues. This diagonalization can be done with a unitary matrix, i.e., we can write the diagonal form of the covariance matrix as
\begin{align}
    \Lambda_Q = U_Q Q U_Q\ad
\end{align}
for some unitary matrix $U_Q$ whose columns correspond to the principal components.

\subsection{PCA without centering}

\begin{figure}[t]
    \centering
    \includegraphics[width =\columnwidth]{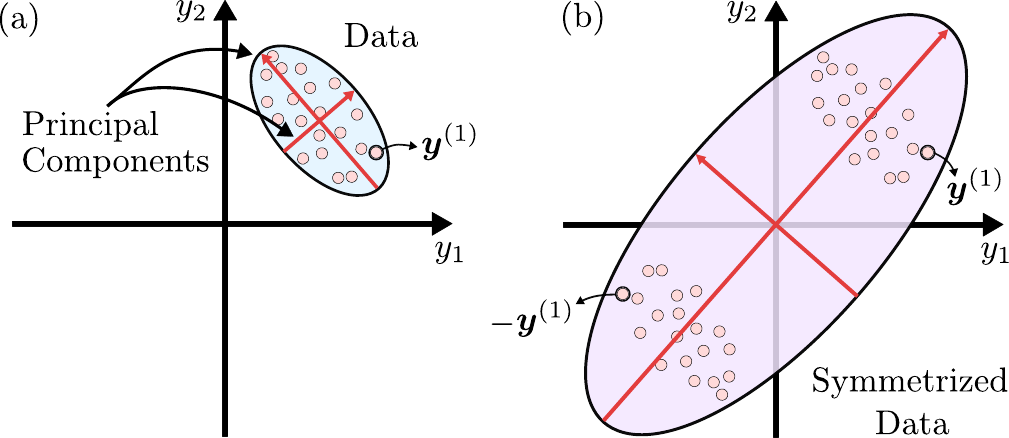}
    \caption{\textbf{Comparison of standard PCA and ``PCA without centering''}. (a) Standard PCA involves diagonalizing $Q$, with the interpretation of fitting an ellipsoid to the data. (b) So-called ``PCA without centering'' can be interpreted as PCA on a symmetrized version of the dataset, e.g., where the dataset is duplicated with the duplicate copy having a minus sign applied. ``PCA without centering'' then corresponds to fitting an ellipsoid to the symmetrized dataset. (See Eq.~\eqref{eqn_data_symmetrized} in the text below for an example of a symmetrized dataset.) }
    \label{fig:symmetrization}
\end{figure}

Let us elaborate here on what is often called ``PCA without centering''. Instead of considering the covariance matrix, one can consider the matrix of (uncentered) second moments, denoted $T$, whose matrix elements are given by:
\begin{align}
    T_{jk}  &= \sum_{i=1}^{N} p^{(i)} y^{(i)}_j  \overline{y^{(i)}_k} \,.
    \label{eqnTjk}
\end{align}
Analogous to standard PCA, one can then diagonalize the $T$ matrix with a unitary matrix $U_T$, to find the diagonal form:
\begin{align}
    \Lambda_T = U_T T U_T\ad\,.
\end{align}
So-called ``PCA without centering'' then involves using the eigenvalues and eigenvectors of $T$ in place of those of $Q$, keeping only the principal components of $T$.

As we will see below, our results indicate that ``PCA without centering'' is relevant to our proposed approach to quantum PCA. Hence, it is important for us to review what is known about this subroutine. 

First of all, the name ``PCA without centering'' is questionable, and we believe an alternative name could be more appropriate. Diagonalizing an uncentered matrix of second moments is not necessarily PCA, and hence the term ``PCA without centering'' is not entirely accurate. Our results below give a different perspective. We argue that ``PCA without centering'' actually corresponds to performing PCA on a symmetrized version of the dataset. So it could be called PCA with data symmetrization. See Fig.~\ref{fig:symmetrization} for an illustration of this, and also see Eq.~\eqref{eqn_data_symmetrized} below for an example of a symmetrized dataset.

Regardless of the terminology, there remains the intriguing question  of how PCA with data symmetrization is related to standard PCA. This question has been studied in detail by Cadima and Jolliffe~\cite{cadima2009relationships}. Their comprehensive work on this topic is extremely useful for our purposes. Therefore, let us review some of their findings here. 

The overall conclusion of Cadima and Jolliffe was that standard PCA and PCA with data symmetrization can be rigorously proven to be closely related. In particular, the eigenvalues and eigenvectors obtained with one method are often similar to those of the other method. For example, they showed that the eigenvalues are interlaced:
\begin{equation}\label{eqn_evalue_interlacing}
   t_d \geq q_d \geq t_{d-1} \geq q_{d-1}\geq ... \geq t_1 \geq q_1\,,
\end{equation}
where $\{t_j\}_{j=1}^d$ and $\{q_j\}_{j=1}^d$ are the eigenvalues of $T$ and $Q$, respectively, listed in non-decreasing order. They also provide sufficient conditions for the sets of eigenvectors of $T$ and $Q$ to perfectly match. For example, one such sufficient condition is when either $T$ or $Q$ has an eigenvector that is identical to the normalized mean vector: $\vec{\mu} /\|\vec{\mu}\|$. It is often the case that the first eigenvector of $T$ is almost co-linear with $\vec{\mu}$, and consequently it is often observed that the eigenvectors of $T$ and $Q$ are quite similar. We refer the reader to Ref.~\cite{cadima2009relationships} for additional theoretical results that establish a close connection between standard PCA and ``PCA without centering''. 

Two minor differences between the setting considered by Cadima and Jolliffe and the setting we consider is that they restrict to real random variables ($y^{(i)}_j \in \mathbb{R}$) and uniform probability distributions over data points ($p^{(i)}=1/N$ for all $i$). In our results (see Sec.~\ref{sct_theorems}), we allow for complex random variables ($y^{(i)}_j \in \mathbb{C}$) and non-uniform probability distributions ($p^{(i)}$ arbitrary). Hence we consider a slightly more general setting than Cadima and Jolliffe.

\section{Definitions}
\label{sec:definitions}

Before stating our results let us first define some notation. 

\subsection{Datasets}

Let us consider a dataset in the form of a set of normalized quantum states $\{\ket{\psi^{(i)}}\}$. For a classical dataset, we can imagine that the $y^{(i)}$ vectors (defined above) are encoded in the $\ket{\psi^{(i)}}$ states through an amplitude encoding procedure. There is a large body of literature on amplitude encoding~\cite{grover2000synthesis,grover2002creating,plesch2011quantum,schuld2018supervised,sanders2019black}, including near-term approaches~\cite{nakaji2021approximate,marin2021quantum,zoufal2019quantum}, and hence we refer the reader to this literature. Hence, in what follows, we can restrict to datasets composed of quantum states.

We denote a dataset of statevectors as
\begin{equation}\label{eqn_dataset_pure1}
    D_{\ket{\psi}} = \{\ket{\psi^{(i)}}\}_{i=1}^N\,,
\end{equation}
and for this dataset, we denote the corresponding dataset of density matrices as
\begin{equation}\label{eqn_dataset_pure2}
    D_{\rho} = \{\rho^{(i)}\}_{i=1}^N = \{\dya{\psi^{(i)}}\}_{i=1}^N\,.
\end{equation}

We note that, more generally, one could have a dataset of mixed states, which we can also denote as $D_{\rho} = \{\rho^{(i)}\}_{i=1}^N$. We give a full treatment of mixed-state datasets in Appendix~\ref{Appendix_mixedstates}. There, we discuss how our results apply to an effective dataset that is constructed from the pure-state decompositions of each mixed state $\rho^{(i)}$. We refer the reader to that appendix for further discussion.

\subsubsection{Classical and quantum datasets}

Some of our theoretical results apply to classical datasets, while others apply to quantum datasets. Hence, let us define these terms here. 

We use the term \textit{classical dataset} to refer to a dataset that is stored on a classical computer or classical device. For classical datasets, the information about global phases for each datapoint is preserved. In other words, multiplicative factors applied to each datapoint have a non-trivial effect on the dataset. This is an important point that distinguishes classical datasets from quantum datasets.

We use the term \textit{quantum dataset} to refer to a dataset that is stored on a quantum computer or quantum device. For quantum datasets, the information about global phases for each datapoint is lost or erased. In other words, multiplicative factors applied to each datapoint have a trivial (or non-physical) effect on the dataset.

\subsection{Ensembles}

We also introduce the notion of ensembles. Ensembles of quantum states are commonly used in quantum information theory~\cite{nielsen2000quantum}. Ensembles include datapoints and their associated probabilities $P = \{p^{(i)}\}_{i=1}^N$. A statevector ensemble is denoted as:
\begin{equation}\label{eqn_ensemble_psi}
    E_{\ket{\psi}} = \{p^{(i)},\ket{\psi^{(i)}}\}_{i=1}^N\,,
\end{equation}
and the associated ensemble of density matrices is
\begin{equation}\label{eqn_ensemble_rho}
    E_{\rho} = \{p^{(i)},\rho^{(i)}\}_{i=1}^N = \{p^{(i)},\dya{\psi^{(i)}}\}_{i=1}^N\,.
\end{equation}

\subsection{Mapping statevectors to density matrices}

For convenience, we denote the mapping that takes a statevector to a density matrix as:
\begin{equation}
    \PC(\ket{\psi}) = \dya{\psi},
\end{equation}
and we call this the outer product mapping, since $\dya{\psi}$ is the outer product. 

With a slight abuse of notation, we can act with this map on a statevector dataset to get the corresponding density matrix dataset. Let $D_{\ket{\psi}} = \{\ket{\psi^{(i)}}\}_{i=1}^N$ be a statevector dataset. Then the corresponding density matrix dataset is:
\begin{equation}
    \PC(D_{\ket{\psi}})  = \{\PC(\ket{\psi^{(i)}})\}_{i=1}^N = \{\dya{\psi^{(i)}}\}_{i=1}^N\,.
\end{equation}

Similarly, we can act with this map on a statevector ensemble to get the corresponding density matrix ensemble. Let $E_{\ket{\psi}} = \{p^{(i)},\ket{\psi^{(i)}}\}_{i=1}^N$ be a statevector ensemble.  Then the corresponding density matrix ensemble is:
\begin{equation}
    \PC(E_{\ket{\psi}})  = \{p^{(i)}, \PC(\ket{\psi^{(i)}})\}_{i=1}^N = \{p^{(i)}, \dya{\psi^{(i)}}\}_{i=1}^N\,.
\end{equation}

We remark that $\PC$ is not an invertible map, since global phase information is lost via the outer product. 

Nevertheless, one can think about the set of statevectors that are consistent with a given density matrix. We borrow terminology from the open-quantum-system literature and refer to a particular choice of statevector (for a given density matrix) as an \textit{unraveling}~\footnote{The open-quantum-system literature refers to a stochastic statevector time evolution as an unraveling of the master equation for the density matrix.}. For example, we say that $e^{i\phi}\ket{\psi}$ is a particular unraveling of the density matrix $\dya{\psi}$. We will also use this same language for ensembles. That is, we will say that $E_{\ket{\psi}}$ is an unraveling of $E_{\rho}$, if it holds that $\PC(E_{\ket{\psi}})= E_{\rho}$.

\subsection{Ensemble average density matrix}

Imagine a simple protocol whereby one samples from the probability distribution $\{p^{(i)}\}$, and if outcome $i$ occurs then one prepares the state $\ket{\psi^{(i)}}$ on a quantum device. This is the protocol previously depicted in Fig.~\ref{fig:preparation}. The result of this protocol is to effectively prepare the state:
\begin{align}\label{eqn_EADM}
\overline{\rho} = \sum_{i=1}^N p^{(i)} \dya{\psi^{(i)}}
\end{align}
which we call the ensemble average density matrix for the ensemble $E_{\ket{\psi}} = \{p^{(i)},\ket{\psi^{(i)}}\}_{i=1}^N$.

\section{Theoretical Results}\label{sct_theorems}

We now proceed to state our theoretical results. We emphasize that, while these results are not mathematically deep, they are conceptually non-trivial and technologically important. 

We also note that the extension of our results to mixed-state datasets is given in Appendix~\ref{Appendix_mixedstates}. In that appendix, we argue that our theoretical results apply to such datasets provided that we consider the pure states that decompose each mixed state as datapoints.

\subsection{Results for classical datasets}

Let us first consider the case of classical datasets defined in Sec.~\ref{sec:definitions}.

\subsubsection{Centered classical datasets}

We begin by considering centered classical datasets. Such datasets have $\vec{\mu} = \vec{0}$, i.e., $\mu_j = 0$ for all $j$. The following proposition gives a simple equality in this case.

\begin{proposition}\label{prop_centered}
Consider a classical dataset of pure states $D_{\ket{\psi}} = \{\ket{\psi^{(i)}}\}_{i=1}^N$. With the ensemble denoted as $E_{\ket{\psi}} = \{p^{(i)},\ket{\psi^{(i)}}\}_{i=1}^N$, the corresponding ensemble average density matrix in \eqref{eqn_EADM} is given by
\begin{equation}
    \overline{\rho}= Q 
\end{equation}
if the dataset is centered, i.e., if $\vec{\mu} = \vec{0}$.
\end{proposition}
\begin{proof}
The proof follows as a special case of Prop.~\ref{prop_uncentered_general} below (i.e., by setting $\vec{\mu} = \vec{0}$ in Prop.~\ref{prop_uncentered_general}). 
\end{proof}

Proposition~\ref{prop_centered} is already useful, as it identifies a method for preparing the covariance matrix for centered datasets, i.e., the method shown in Fig.~\ref{fig:preparation}.

\subsubsection{Uncentered classical datasets}

The following proposition gives the general relationship between $\overline{\rho}$ and $Q$, regardless of whether the data is centered or uncentered.  We note that Eq.~\eqref{eqn_lemma1} below is mathematically related to (although conceptually different from) a result obtained in Ref.~\cite{cadima2009relationships}.

\begin{proposition}
\label{prop_uncentered_general}
Consider a classical dataset of pure states $D_{\ket{\psi}} = \{\ket{\psi^{(i)}}\}_{i=1}^N$. With the ensemble denoted as $E_{\ket{\psi}} = \{p^{(i)},\ket{\psi^{(i)}}\}_{i=1}^N$, the corresponding ensemble average density matrix in \eqref{eqn_EADM} is given by
\begin{equation}\label{eqn_lemma1}
    \overline{\rho}= Q + M
\end{equation}
where $Q$ is the covariance matrix for the ensemble $E_{\ket{\psi}}$. Here, $M = \vec{\mu} \vec{\mu}\ad$ has matrix elements $M_{jk} = \mu_j \overline{\mu_k}$, with $\vec{\mu} = \{\mu_1,...,\mu_d\}^T$ being the vector of mean values for the $d$ features of the data.
\end{proposition}
\begin{proof}
Let us expand each state in terms of its amplitudes in the standard basis: 
\begin{align}
\ket{\psi^{(i)}} = \sum_{j=1}^d y^{(i)}_j \ket{j} \,.
\end{align}
Here we denote the amplitudes as $y^{(i)}_j$ because we assume that one has performed amplitude encoding of the data, and we recall that our notation above used $y^{(i)}_j$ to denote the $j$-th feature value of the $i$-th data point.

Then we have that
\begin{align}
\dya{\psi^{(i)}} = \sum_{j,k=1}^d y^{(i)}_j \overline{y^{(i)}_k}\dyad{j}{k} \,. 
\end{align}
Hence we can write
\begin{align}
\overline{\rho} &= \sum_{i=1}^N \sum_{j,k=1}^d p^{(i)} y^{(i)}_j \overline{y^{(i)}_k}\dyad{j}{k} \\
&= \sum_{j,k=1}^d \overline{\rho}_{jk}\dyad{j}{k}\,.
\end{align}
Here the matrix elements of $\overline{\rho}$ are given by:
\begin{align}\label{eqn_rhojk_Tjk}
\overline{\rho}_{jk} = \sum_{i=1}^N p^{(i)} y^{(i)}_j \overline{y^{(i)}_k} = T_{jk}\,,
\end{align}
where $T_{jk}$ was defined in Eq.~\eqref{eqnTjk}. Using Equation~\eqref{eqnKjk}, we can rewrite~\eqref{eqn_rhojk_Tjk} in terms of the covariance matrix elements as follows:
\begin{align}
\overline{\rho}_{jk} = Q_{jk} + \mu_j \overline{\mu_k} \,,
\end{align}
which completes the proof.
\end{proof}

Equation~\eqref{eqn_lemma1} implies that $\overline{\rho}$ and $Q$ are very closely related, even for uncentered datasets. The difference between $\overline{\rho}$ and $Q$ is a positive semi-definite matrix $M$, which has several interesting properties, as stated in the following lemma.

\begin{lemma}\label{lemma_M}
For uncentered datasets, the $d\times d$ matrix $M = \vec{\mu}\vec{\mu}\ad$ has the following properties:
\begin{itemize}
    \item It is Hermitian and positive semi-definite. 
    \item It is rank-one.
    \item Its only non-zero eigenvalue is 
    \begin{equation}
        \lambda_{\vec{\mu}} = \vec{\mu}\ad\vec{\mu}= \|\vec{\mu}\|^2 = \mu_1^*\mu_1 + ...  + \mu_d^*\mu_d
    \end{equation}
    \item The eigenvector associated with this non-zero eigenvalue is $v_{\vec{\mu}} = \vec{\mu}/ \|\vec{\mu}\|$.
\end{itemize}
\end{lemma}
\begin{proof}
The fact that it is Hermitian is obvious since $M\ad = M$. The fact that is rank-one follows from the ability to write $M = \lambda_{\vec{\mu}}\dya{\vec{\mu}}$ as being proportonal to a projector onto a pure quantum state $\ket{\vec{\mu}}$. This also implies that it is positive semi-definite, since it is proportional to a density matrix, with a positive proportionality constant. The fact the $v_{\vec{\mu}} $ is an eigenvector with eigenvalue $\lambda_{\vec{\mu}} = \vec{\mu}\ad\vec{\mu}$ can be directly verified: 
\begin{equation}
    M v_{\vec{\mu}} = \vec{\mu}\vec{\mu}\ad \vec{\mu}/ \|\vec{\mu}\| = (\vec{\mu}\ad \vec{\mu}) \vec{\mu}/ \|\vec{\mu}\| = \lambda_{\vec{\mu}} v_{\vec{\mu}}\,.
\end{equation}
\end{proof}

We are now in a position to relate the eigenvalues of $\overline{\rho}$ and $Q$. Namely, we formulate inequalities that bound the deviation of the eigenvalues of $\overline{\rho}$ from those of $Q$. The proof of the following result relies on Prop.~\ref{prop_uncentered_general}, Lemma~\ref{lemma_M}, and Weyl's theorem as stated in Appendix~\ref{Appendix_Weyl}.

\begin{proposition}\label{prop_spectrum_evalues}
Consider a classical dataset of pure states $D_{\ket{\psi}} = \{\ket{\psi^{(i)}}\}_{i=1}^N$, with the ensemble denoted as $E_{\ket{\psi}} = \{p^{(i)},\ket{\psi^{(i)}}\}_{i=1}^N$. Let $\{r_j\}_{j=1}^d$ and $\{q_j\}_{j=1}^d$, respectively, be the eigenvalues of $\overline{\rho}$ and $Q$ listed in non-decreasing order. Then the eigenvalues of $\overline{\rho}$ are interlaced with the eigenvalues of $Q$, as follows:
\begin{equation}\label{eqn_eigenvalues_rjqj_interlacing}
    r_d\geq q_d \geq r_{d-1} \geq q_{d-1} \geq ... \geq r_1 \geq q_1\,.
\end{equation}
In addition, for each $j$, the following bound holds
\begin{equation}\label{eqn_eigenvalues_rjqj_interval}
    q_j \leq r_j \leq q_j + \|\vec{\mu}\|^2\,,
\end{equation}
where $\vec{\mu}$ is the vector of mean values.
\end{proposition}
\begin{proof}
Recall from Prop.~\ref{prop_uncentered_general} that we have $\overline{\rho} = Q + M$. Using this result, we can then apply Weyl's theorem for the eigenvalues of Hermitian matrices, noting that $\overline{\rho}$, $Q$, and $M$ are all Hermitian. Appendix~\ref{Appendix_Weyl} gives the general statement of Weyl's theorem, which relates the eigenvalues of two Hermitian matrices $A$ and $B$ to those of $A+B$. We refer the reader to Ref.~\cite{horn2012matrix} for additional details on Weyl's theorem. In Appendix~\ref{Appendix_Weyl}, we specialize Weyl's theorem to the case where $A$ is an arbitrary Hermitian matrix and $B$ is a rank-one Hermitian matrix whose only non-zero eigenvalue is $\lambda_B >0$, to obtain:
\begin{align}
        \lambda_j(A+B) &\leq \lambda_{j}(A) + \lambda_B\quad\text{for }j = 1, ..., d\label{eqn_evalue_ordering1}\\
        \lambda_j(A+B) &\leq \lambda_{j+1}(A)\quad\text{for }j = 1, ..., d-1 \label{eqn_evalue_ordering2}\\
        \lambda_j(A) &\leq \lambda_{j}(A+B)\quad\text{for }j = 1, ..., d \label{eqn_evalue_ordering3}\,.
\end{align}
In this set of inequalities, the respective eigenvalues of $A$ and $A+B$ are denoted $\{\lambda_j(A)\}_{j=1}^d$ and $\{\lambda_j(A+B)\}_{j=1}^d$, and these eigenvalues are listed in non-decreasing order. Let us apply the above inequalities by choosing $A= Q$, $B=M$, and $A+B = \overline{\rho}$. In this case, $\lambda_j(A+B)$ becomes $r_j$, and $\lambda_j(A)$ becomes $q_j$. Moreover, by invoking Lemma~\ref{lemma_M}, the non-zero eigenvalue of $B$, $\lambda_B$, corresponds to $\|\vec{\mu}\|^2$. Hence, the above inequalities become:
\begin{align}
        r_j &\leq q_j + \|\vec{\mu}\|^2\quad\text{for }j = 1, ..., d\label{eqn_evalue_ordering4}\\
        r_j &\leq q_{j+1}\quad\text{for }j = 1, ..., d-1 \label{eqn_evalue_ordering5}\\
        q_j &\leq r_j\quad\text{for }j = 1, ..., d \label{eqn_evalue_ordering6}\,.
\end{align}
Combining \eqref{eqn_evalue_ordering5} and \eqref{eqn_evalue_ordering6} gives the result in \eqref{eqn_eigenvalues_rjqj_interlacing}. Also, combining \eqref{eqn_evalue_ordering4} and \eqref{eqn_evalue_ordering6} gives the result in \eqref{eqn_eigenvalues_rjqj_interval}.
\end{proof}

Proposition~\ref{prop_spectrum_evalues} implies that the eigenvalues of $\overline{\rho}$ and $Q$ can never deviate too much from each other. For example, for the $j$th eigenvalue, the deviation $\delta_j :=r_j - q_j $ is upper bounded as follows:
\begin{equation}
    \delta_j \leq \min\{\|\vec{\mu}\|^2, (q_{j+1}-q_j)\}\,.
\end{equation}
In this sense, the eigenvalues of $\overline{\rho}$ form a good approximation for the eigenvalues of $Q$.

We remark that Prop.~\ref{prop_spectrum_evalues} is more general than the main result in Ref.~\cite{cadima2009relationships}. Specifically, Prop.~\ref{prop_spectrum_evalues} holds for arbitrary probability distributions $\{p^{(i)}\}$ and for complex random variables $y^{(i)}_j$, whereas Ref.~\cite{cadima2009relationships} restricted to uniform probability distributions and real random variables.

We also note that the proof technique used to prove the main result in Ref.~\cite{cadima2009relationships} is only valid for uniform probability distributions. Therefore, to prove Prop.~\ref{prop_spectrum_evalues}, we could not simply use the proof technique in Ref.~\cite{cadima2009relationships}. Rather, we required a novel proof technique, and this involved using Weyl's theorem.

Let us now relate the eigenvectors of $\overline{\rho}$ and $Q$. In general, for a matrix $A$ and a state $\ket{\psi}$, we can quantify how far $\ket{\psi}$ is from being an eigenvector of $A$, as follows. Let us define the unnormalized state 
\begin{equation}\label{eqn_deltavector}
    \ket{\delta_{A,\ket{\psi}}} = A\ket{\psi} - \mte{\psi}{A}\ket{\psi}\,.
\end{equation}
Note that $\ket{\delta_{A,\ket{\psi}}}$ is the zero vector whenever $\ket{\psi}$ is an eigenvector of $A$. The norm of the $\ket{\delta_{A,\ket{\psi}}}$ vector quantifies how far $\ket{\psi}$ is from being an eigenvector of $A$. Therefore we define the eigenvector error as follows:
\begin{equation}\label{eqn_norm_delta}
    e_{A,\ket{\psi}} = \ip{\delta_{A,\ket{\psi}}}{\delta_{A,\ket{\psi}}}\,.
\end{equation}
Consider the following proposition that quantifies the eigenvector error.
\begin{proposition}\label{prop_spectrum_evector_error}
Consider a classical dataset of pure states $D_{\ket{\psi}} = \{\ket{\psi^{(i)}}\}_{i=1}^N$, with the ensemble denoted as $E_{\ket{\psi}} = \{p^{(i)},\ket{\psi^{(i)}}\}_{i=1}^N$. Let $\ket{Q_j}$ be an eigenvector of the covariance matrix $Q$. For this state, the eigenvector error for the matrix $\overline{\rho}$ is:
\begin{equation}\label{eqn_evectorerror1_prop}
   e_{\overline{\rho},\ket{Q_j}} = \|\vec{\mu}\|^4 |\ip{v_{\vec{\mu}}}{Q_j}|^2 (1- |\ip{v_{\vec{\mu}}}{Q_j}|^2)\,.
\end{equation}
Similarly, let $\ket{R_j}$ be an eigenvector of the ensemble average density matrix $\overline{\rho}$. For this state, the eigenvector error for the matrix $Q$ is given by the same expression:
\begin{equation}\label{eqn_evectorerror2_prop}
   e_{Q,\ket{R_j}} = \|\vec{\mu}\|^4 |\ip{v_{\vec{\mu}}}{R_j}|^2 (1- |\ip{v_{\vec{\mu}}}{R_j}|^2)\,.
\end{equation}
\end{proposition}
\begin{proof}
This follows from a direct calculation. One can write:
\begin{align}
    \ket{\delta_{\overline{\rho},\ket{Q_j}}} &= \overline{\rho}\ket{Q_j} - \mte{Q_j}{\overline{\rho}}\ket{Q_j}\\
    &= \overline{\rho}\ket{Q_j} - \ket{Q_j}\mte{Q_j}{\overline{\rho}}\\
    &=\Pi^{\perp}\overline{\rho} \ket{Q_j}
\end{align}
where $\Pi^{\perp} = \id - \dya{Q_j}$ is the projector onto the orthogonal complement of $\ket{Q_j}$. Next we use $\overline{\rho} = Q+M$ to write
\begin{align}
    \ket{\delta_{\overline{\rho},\ket{Q_j}}} 
    &=\Pi^{\perp}(Q+M) \ket{Q_j}\\
    &=\Pi^{\perp} M  \ket{Q_j}\,,
\end{align}
which follows from the fact that $\Pi^{\perp} Q  \ket{Q_j} = 0$, since $\ket{Q_j}$ is an eigenvector of $Q$. The eigenvector error is then:
\begin{align}
    e_{\overline{\rho},\ket{Q_j}}&= \bra{Q_j}M \Pi^{\perp} M \ket{Q_j} \\
    &=\bra{Q_j}M^2 \ket{Q_j} - \bra{Q_j}M \dya{Q_j} M \ket{Q_j}
\end{align}
Combining this final expression with $M = \|\vec{\mu}\|^2 \dya{v_{\vec{\mu}}}$ then gives the desired result in \eqref{eqn_evectorerror1_prop}. The above proof can be rewritten analogously to derive \eqref{eqn_evectorerror2_prop}.
\end{proof}

From this proposition, we see that an eigenvector for $\overline{\rho}$ ($Q$) is an eigenvector of $Q$ ($\overline{\rho}$) if and only if at least one of the following conditions is satisfied:
\begin{itemize}
    \item The eigenvector is orthogonal to the mean vector.
    \item The eigenvector is colinear with the mean vector.
    \item The mean vector is the zero vector (i.e, the dataset is centered).
\end{itemize}
Hence, the only way for an eigenvector to not be shared between $\overline{\rho}$ and $Q$ is if the dataset is uncentered and the eigenvector has partial (but not complete) overlap with the mean vector. Because of Prop.~\ref{prop_spectrum_evector_error}, one can see that it is often the case that eigenvectors are approximately shared between the two matrices, $\overline{\rho}$ and $Q$. 

In addition, Eq.~\eqref{eqn_evectorerror2_prop} can be used as a diagnostic tool for ``PCA without centering''. Specifically, this equation can be used to verify the quality of an eigenvector obtained from diagonalizing $\overline{\rho}$. One can simply calculate the overlap of that eigenvector with the mean vector and then use \eqref{eqn_evectorerror2_prop} to quantify the eigenvalue error for the true covariance matrix $Q$.

Now let us consider both the eigenvalues and eigenvectors. The following proposition gives a sufficient condition for the both the eigenvalues and eigenvectors of $\overline{\rho}$ and $Q$ to perfectly match. As we will see in our numerical implementations, this sufficient condition is often satisfied. Namely, it is often the case that the first principal component of $\overline{\rho}$ (i.e., the eigenvector associated with the largest eigenvalue) is very close to being the normalized mean vector. Hence, the following proposition is relevant to practical scenarios of interest.

\begin{proposition}\label{prop_spectrum_evectors}
Consider a classical dataset of pure states $D_{\ket{\psi}} = \{\ket{\psi^{(i)}}\}_{i=1}^N$, with the ensemble denoted as $E_{\ket{\psi}} = \{p^{(i)},\ket{\psi^{(i)}}\}_{i=1}^N$. If one of the eigenvectors of $\overline{\rho}$ or $Q$ is the normalized mean vector $\vec{\mu}/\|\vec{\mu}\|$, then
\begin{itemize}
    \item  $\overline{\rho}$ and $Q$ have a common set of eigenvectors. Any set of eigenvectors for $\overline{\rho}$ is also a valid set of eigenvectors for $Q$, and vice versa.
    \item  The $d-1$ eigenvalues that are not associated with $\vec{\mu}/\|\vec{\mu}\|$ are shared. Letting $r_{\hat{j}}$ and $q_{\hat{j}}$ be the eigenvalues associated with  $\vec{\mu}/\|\vec{\mu}\|$, then $r_{\hat{j}} - q_{\hat{j}} = \|\vec{\mu}\|^2$.
\end{itemize}
\end{proposition}
\begin{proof}
Denote the normalized mean vector as the quantum state $\ket{v_{\vec{\mu}}} = \vec{\mu}/ \|\vec{\mu}\|$. Consider a spectral decomposition of $Q$ given by
\begin{equation}
    Q = \sum_{j\neq \hat{j}} q_j \dya{Q_j} + q_{\hat{j}} \dya{v_{\vec{\mu}}}
\end{equation}
where the $\ket{Q_j}$ states are orthogonal eigenvectors (and also orthogonal to $\ket{v_{\vec{\mu}}}$). Then, we invoke the equation $\overline{\rho} = Q+M$ to obtain
\begin{align}
    \overline{\rho} &= Q + \|\vec{\mu}\|^2 \dya{v_{\vec{\mu}}}\\
    &= \sum_{j\neq \hat{j}} q_j \dya{Q_j} + (q_{\hat{j}}+ \|\vec{\mu}\|^2) \dya{v_{\vec{\mu}}}\,.\label{eqn_spectrum_inheritance}
\end{align}
Note that this is also a spectral decomposition of $\overline{\rho}$. Hence, we see that $\overline{\rho}$ inherits the same eigenvectors as those of $Q$, and all of the eigenvalues are also the same except for the $\hat{j}$th eigenvalue, which is shifted by $\|\vec{\mu}\|^2$. Finally, note that one can apply the exact same argument in the reverse direction, where one first starts with a spectral decomposition of $\overline{\rho}$, and then one derives the corresponding spectral decomposition of $Q$.  
\end{proof}

Thusfar, we have related the eigenvalues and eigenvectors of $\overline{\rho}$ and $Q$. We remark that one can also relate the diagonal elements of $\overline{\rho}$ and $Q$, and we formally state this in Appendix~\ref{Appendix_diagonalelements}.

\subsection{Results for quantum datasets}

\begin{figure}[t]
    \centering
    \includegraphics[width = \columnwidth]{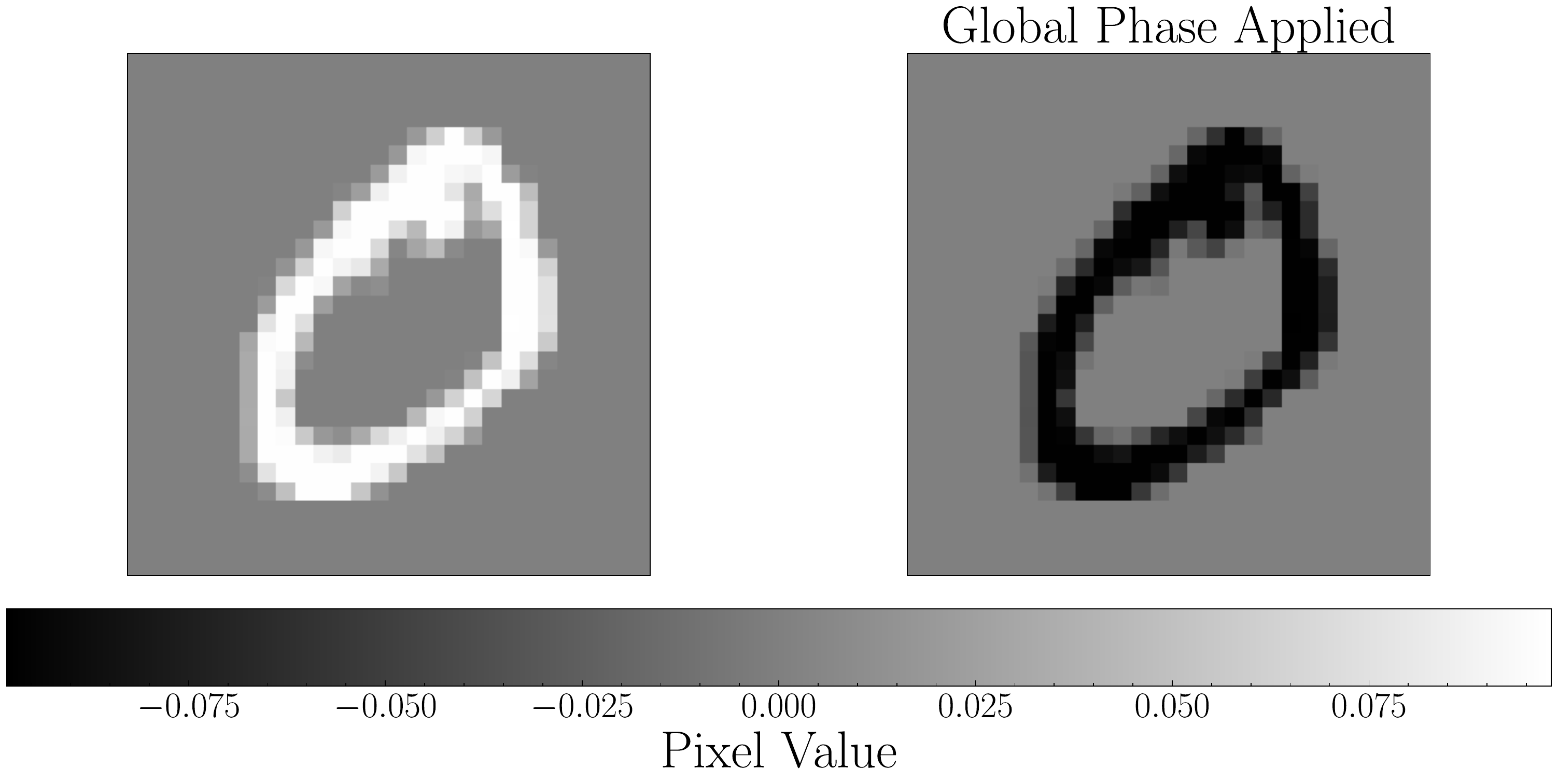}
    \caption{\textbf{Impact of global phase on classical data.} Unlike quantum states, classical datapoints are impacted by applying a global phase. Here we give a visual demonstration of this, where multiplying by a minus sign inverts the color of a handwritten digit from white to black.}
    \label{fig:globalphase_digit}
\end{figure}

Classical datapoints are affected by application of a global phase, such as multiplying by a minus sign as in Fig.~\ref{fig:globalphase_digit}. One can see in this case that the minus sign inverts the color of the image, from a white digit to a black digit.

However, quantum states are invariant under application of a global phase to the state vector. This can be seen from the fact that the density matrix is unaffected by a global phase applied to the state vector:
\begin{align}
    \ket{\psi}&\rightarrow e^{i\phi}\ket{\psi}\notag\\
    \dya{\psi}&\rightarrow e^{i\phi}e^{-i\phi}\dya{\psi} = \dya{\psi}\,.
\end{align}
This global phase symmetry implies that global phases have no physical effect or physical manifestation. We now discuss how global phase symmetry can allow us to assume that all datasets that are prepared on quantum devices admit a description that is centered. More precisely, there are multiple statevector descriptions for a given dataset of density matrices, and there always exist some statevector descriptions that are centered.

We first state the following lemma.
\begin{lemma}\label{lemma2}
Consider a quantum dataset of pure states. In this case, the global phase information is lost. Hence the dataset can be described by an ensemble of density matrices $E_{\rho} = \{p^{(i)},\dya{\psi^{(i)}}\}_{i=1}^N$. Then, there always exists a statevector ensemble $E_{\ket{\tilde{\psi}}} = \{\tilde{p}^{(j)},\ket{\tilde{\psi}^{(j)}}\}_{j=1}^{\tilde{N}}$ that satisfies the following conditions:
\begin{itemize}
    \item $E_{\ket{\tilde{\psi}}}$ physically corresponds to the aforementioned $E_{\rho}$, in the sense that applying the outer product mapping leads to $\PC(E_{\ket{\tilde{\psi}}}) = E_{\rho}$. In other words, $E_{\ket{\tilde{\psi}}}$ is an unraveling of $E_{\rho}$.
    \item $E_{\ket{\tilde{\psi}}}$ is centered, i.e., the mean value of all features is zero.
\end{itemize}
\end{lemma}
\begin{proof}
We will prove this lemma by construction, i.e., by constructing an $E_{\ket{\tilde{\psi}}}$ that satisfies the required conditions.

Let $D_{\rho} = \{\dya{\psi^{(i)}}\}_{i=1}^N$, and let $D_{\ket{\psi}} = \{\ket{\psi^{(i)}}\}_{i=1}^N$ be a particular unraveling of $D_{\rho}$.  

Now let us consider a symmetrized dataset composed of $\tilde{N}=2N$ datapoints that involves appending the set $(-D_{\ket{\psi}})$ onto the end of the set $D_{\ket{\psi}}$. We write the overall dataset as 
\begin{align}\label{eqn_data_symmetrized}
    D_{\ket{\tilde{\psi}}} &= \{\ket{\psi^{(1)}},...,\ket{\psi^{(N)}},-\ket{\psi^{(1)}},...,-\ket{\psi^{(N)}}    \}\,.
\end{align}
For the probabilities associated with this dataset, we use
\begin{align}
    \tilde{P} &= \frac{1}{2}\{ p^{(1)},..., p^{(N)},p^{(1)},..., p^{(N)}\}\,.
\end{align}
Using the previous two equations, we can write the overall ensemble as:
\begin{align}\label{eqn_Epsi2N}
    E_{\ket{\tilde{\psi}}}  =& \bigg\{ (\frac{p^{(1)}}{2},\ket{\psi^{(1)}}),..., (\frac{p^{(N)}}{2},\ket{\psi^{(N)}}),\notag\\
    &(\frac{p^{(1)}}{2},-\ket{\psi^{(1)}}),...,(\frac{p^{(N)}}{2},-\ket{\psi^{(N)}})\bigg\}\,.
\end{align}
Let us note that this ensemble, $E_{\ket{\tilde{\psi}}}$, is centered. One can see this by computing the mean values for each feature as:
\begin{align}\label{eqn_Epsi2N_centered}
    \mu_j = \sum_{i=1}^{N} \frac{p^{(i)}}{2} y^{(i)}_j +\sum_{i=1}^N \frac{p^{(i)}}{2} (- y^{(i)}_j) =0
\end{align}
Here we used the fact that the feature values for $-\ket{\psi^{(i)}}$ are the negatives of the feature values for $\ket{\psi^{(i)}}$. Hence we have shown that $E_{\ket{\tilde{\psi}}}$ is centered.

Now we just need to show that $E_{\ket{\tilde{\psi}}}$ is an unraveling of $E_{\rho}$. In other words we need to show that $\PC(E_{\ket{\tilde{\psi}}}) = E_{\rho}$ where $\PC$ is the outer product mapping. Applying $\PC$ gives:
\begin{align}\label{eqn_PCEpsi2n}
    \PC(E_{\ket{\tilde{\psi}}})  =& \bigg\{ (\frac{p^{(1)}}{2},\dya{\psi^{(1)}}),..., (\frac{p^{(N)}}{2},\dya{\psi^{(N)}}),\notag\\
    &(\frac{p^{(1)}}{2},\dya{\psi^{(1)}}),...,(\frac{p^{(N)}}{2},\dya{\psi^{(N)}})\bigg\}\,.
\end{align}
Note that there are $N$ datapoints that appear twice in this ensemble with the same probability. These redundant datapoints can be aggregated, with their propabilities summed together. After we aggregate these datapoints together, we see that the ensemble $\PC(E_{\ket{\tilde{\psi}}})$ is equivalent to the ensemble $E_{\rho} = \{p^{(i)},\dya{\psi^{(i)}}\}_{i=1}^N$. This proves the desired result.
\end{proof}

In what follows, it will help if we use more explicit notation. We will use $\overline{\rho}(E_{\rho})$, instead of $\overline{\rho}$, to indicate that the ensemble average density matrix is a function of the ensemble $E_{\rho}$ in \eqref{eqn_ensemble_rho}. We will also use $Q(E_{\ket{\psi}})$, instead of $Q$, to indicate that the covariance matrix is a function of the statevector ensemble $E_{\ket{\psi}}$ in \eqref{eqn_ensemble_psi}. 

With the previous lemma in hand, we can state the following proposition for quantum datasets.
\begin{proposition}\label{prop_quantumdataset}
Consider a quantum dataset of pure states. In this case, the global phase information is lost. Hence the dataset can be described by an ensemble of density matrices $E_{\rho} = \{p^{(i)},\dya{\psi^{(i)}}\}_{i=1}^N$. Then, there always exists a statevector ensemble $E_{\ket{\tilde{\psi}}} = \{\tilde{p}^{(j)},\ket{\tilde{\psi}^{(j)}}\}_{j=1}^{\tilde{N}}$ that satisfies the following conditions:
\begin{itemize}
    \item $E_{\ket{\tilde{\psi}}}$ physically corresponds to the aforementioned $E_{\rho}$, in the sense that applying the outer product mapping leads to $\PC(E_{\ket{\tilde{\psi}}}) = E_{\rho}$. In other words, $E_{\ket{\tilde{\psi}}}$ is an unraveling of $E_{\rho}$. 
    \item The covariance matrix $Q(E_{\ket{\tilde{\psi}}})$ for $E_{\ket{\tilde{\psi}}}$ is equal to the ensemble average density matrix for $E_{\rho}$:
    \begin{equation}
    Q(E_{\ket{\tilde{\psi}}}) = \overline{\rho}(E_{\rho})\,.
\end{equation}
\end{itemize}
\end{proposition}
\begin{proof}
We will prove this by constructing an $E_{\ket{\tilde{\psi}}}$ that satisfies the required conditions. Specifically, we will use the same $E_{\ket{\tilde{\psi}}}$ used to prove the previous lemma, with $E_{\ket{\tilde{\psi}}}$ given in Eq.~\eqref{eqn_Epsi2N}.

For this choice of $E_{\ket{\tilde{\psi}}}$, we already showed that it satisfies the first condition, i.e., that it is an unraveling of $E_{\rho}$ (see Eq.~\eqref{eqn_PCEpsi2n}). 

Hence we just need to show the second condition. The second condition essentially follows by combining Prop.~\ref{prop_centered} and Lemma~\ref{lemma2}. In more detail, in Lemma~\ref{lemma2} we showed that $E_{\ket{\tilde{\psi}}}$ is centered, see Eq.~\eqref{eqn_Epsi2N_centered}. Hence $\vec{\mu} = \vec{0}$ for $E_{\ket{\tilde{\psi}}}$. Next we apply Prop.~\ref{prop_centered} to see that
\begin{equation}
    Q(E_{\ket{\tilde{\psi}}}) = \overline{\rho}(E_{\ket{\tilde{\psi}}})\,,
\end{equation}
where $\overline{\rho}(E_{\ket{\tilde{\psi}}})$ is the ensemble average density matrix for $\overline{\rho}(E_{\ket{\tilde{\psi}}})$. One can also see this by applying Eq.~\eqref{eqn_lemma1} and setting $\vec{\mu}=0$. Finally, one can note that the ensemble average density matrix is the same for $E_{\ket{\tilde{\psi}}}$ and $E_{\rho}$,
\begin{equation}
    \overline{\rho}(E_{\ket{\tilde{\psi}}}) = \overline{\rho}(E_{\rho})\,.
\end{equation}
Combining the two previous equations gives the desired result.
\end{proof}

\subsection{Implications of theoretical results}

The implications of our theoretical results are as follows. If you have any dataset of pure states, then their ensemble average density matrix $\overline{\rho}$ can always be interpreted as a covariance matrix, for an appropriately symmetrized (or centered) version of the dataset. 

If the original dataset of pure states lives on a quantum device, then their global phase information has already been erased, and hence one can always interpret this dataset as being centered. This means that the action of extracting the principal eigenvectors of $\overline{\rho}$, for any pure state dataset living on a quantum device, can always be interpreted as performing PCA on this dataset. This result is extremely useful in light of the fact that quantum PCA for quantum datasets could lead to exponential quantum speedup~\cite{cotler2021revisiting,huang2021quantumadvantage}. We therefore believe that Prop.~\ref{prop_quantumdataset} will have technological importance in the quest for quantum advantage, as it provides a simple means for preparing the covariance matrix for quantum datasets.

With that said, for classical datasets (living on classical devices), the global phase information is retained. The action of encoding this dataset into quantum states and preparing $\overline{\rho}$ on a quantum device will necessarily destroy the global phase information. In this case, extracting the principal eigenvectors of $\overline{\rho}$ will only correspond to PCA for a slightly different dataset, i.e., a symmetrized version of the original dataset. This idea was depicted in Fig.~\ref{fig:symmetrization}.

Nevertheless, Props.~\ref{prop_spectrum_evalues}, \ref{prop_spectrum_evector_error}, \ref{prop_spectrum_evectors} show that the spectrum obtained from $\overline{\rho}$ can be quite similar to that obtained from standard PCA. Proposition~\ref{prop_spectrum_evector_error} is a conceptually novel formulation, while Props.~\ref{prop_spectrum_evalues} and \ref{prop_spectrum_evectors} generalize the results in Ref.~\cite{cadima2009relationships} to the cases of complex random variables and non-uniform probability distributions over data points. (This generalization was non-trivial and required a different proof technique than that used in  Ref.~\cite{cadima2009relationships}). In particular, we find that the spectrum of $\overline{\rho}$ is essentially identical to that of the covariance matrix whenever one eigenvector is colinear with the mean vector $\vec{\mu}$. This condition is often satisfied by the first principal eigenvector of $\overline{\rho}$~\cite{cadima2009relationships}. One can see this geometrically in Fig.~\ref{fig:symmetrization}(b), where the principal axis of the ellipse is approximately colinear with the mean vector.

We further investigate these issues, for both classical and quantum datasets, in our numerical implementations in Sec.~\ref{sct_numerics}.

\section{Analysis of the Sampling Overhead}\label{sct_AnalysisOfSampling}
Here we detail the sampling cost incurred by our approach when being used as a subroutine for two near-term algorithms that can be used to diagonalize a density matrix. Namely, we explore the application of our approach within the variational quantum state diagonalization (VQSD)
algorithm \cite{larose2019variational} and the variational quantum state eigensolver (VQSE) algorithm \cite{cerezo2020variational} for performing quantum PCA. In addition, we briefly remark about the relevance of our approach to the original quantum PCA algorithm~\cite{lloyd2014quantum} in Sec.~\ref{sct_originalqPCA}.

In both VQSD and VQSE a cost function is minimized in order to find a circuit that diagonalizes a quantum state $\rho$. This circuit can then be used to return the approximate eigenvalues and eigenvectors of the state. Naturally, in our analysis, we will assume that the state to be diagonalized is the ensemble average density matrix, i.e., $\rho = \overline{\rho}$. After all, we connect $\overline{\rho}$ to the covariance matrix in our work here, and PCA involves diagonalizing the covariance matrix.

In VQSD, the cost function takes the form
\begin{equation}\label{eq:Cost_VQSD}
    \mathcal{C}_\textrm{VQSD}(U(\boldsymbol{\alpha}), \rho) = \Tr(\rho^{2}) - \Tr(\mathcal{Z}((U(\boldsymbol{\alpha})\rho U(\boldsymbol{\alpha})^{\dagger})^{2})\,
\end{equation}
where $\mathcal{Z}(\rho)$ is a quantum channel that dephases $\rho$ in the standard basis.  The cost function vanishes under the condition $\tilde{\rho} = \mathcal{Z}(\tilde{\rho})$, meaning $\tilde{\rho}$ is diagonal in the standard basis. Two copies of the state $\rho$ are required to compute the terms in the VQSD cost. The first term can be computed using the destructive swap test, whereas the second term can be evaluated using the Diagonalized Inner Product (DIP) test. 


For VQSE the cost function has the form
\begin{equation}\label{eq:Cost_VQSE}
    \mathcal{C}_\textrm{VQSE}(V(\boldsymbol{\theta}), \rho) = \Tr( V(\boldsymbol{\theta})\rho V^{\dagger}(\boldsymbol{\theta})H)
\end{equation}
where $H$ is a Hamiltonian that is non-degenerate over its $m$-lowest energy levels, assuming that one wishes to extract the $m$-largest eigenvalues of $\rho$. While this leaves much freedom in choosing $H$, one possible form is 
\begin{equation}\label{eqn_Hvqse}
    H = \mathbb{1} - \sum\limits_{i=1}^{m} q_{i} \ketbra{\boldsymbol{e}_{i}}{\boldsymbol{e}_{i}}\,
\end{equation}
and $q_{i}>0$ (such that $q_{i}>q_{i+1}$) and the $\ket{\boldsymbol{e}_{i}}$ are orthogonal states in the standard basis. In this case only one copy of the state $\rho$ is needed to evaluate the cost function. 

We first consider the case where every state in the dataset of interest is prepared deterministically and the statistics of the probability distribution are reproduced in a classical post processing step. Then we consider the case where we prepare each state by sampling from distribution of the dataset. 

Although the cost functions noted above are global and hence can have trainability issues \cite{cerezo2020cost}, we consider them here due to their simplicity, and we note the analysis we present can easily be extended to local versions of these cost functions.

\subsection{Deterministic state preparation}\label{sct_deterministic}

Let us consider a deterministic state preparation as follows. One can imagine rewriting the cost functions in \eqref{eq:Cost_VQSD} and \eqref{eq:Cost_VQSE} in terms of the states in the ensemble. In others words, we expand the density matrix as in Eq.~\eqref{eqn_EADM} as 
$\rho = \overline{\rho} = \sum_{i=1}^N p^{(i)} \dya{\psi^{(i)}}$, and we insert this expression into the cost functions. Now the cost functions are written entirely in terms of the states $\ket{\psi^{(i)}}$ that compose the ensemble. Hence one can estimate these cost functions using state preparation circuits for the individual $\ket{\psi^{(i)}}$ states.

When considering VQSE, the situation is very simple, as only one copy of the state $\rho$ is necessary to evaluate the cost function. Therefore, one only needs to compute the output from $\mathcal{O}(N)$ circuits at each cost function evaluation. Hence the number of state preparation circuits is linear in $N$, which is a relatively minor overhead.

For the case of VQSD, the cost function is quadratic in $\rho$. Hence the circuits required to evaluate each term in the cost function require two copies of $\rho$. When expanding the cost function in terms of the states that decompose $\rho$ (as in Eq.~\eqref{eqn_EADM}), there are $N^2$ terms in the expansion. Hence, each of the two terms of the VQSD cost in \eqref{eq:Cost_VQSD} requires $N^2$ state preparation circuits, as the first term is evaluated with the destructive swap test and the second term is evaluated with the DIP test. Therefore, the overall VQSD cost function can be evaluated using $\mathcal{O}(N^{2})$ state preparation circuits.

In summary, only a small amount of overhead in $N$ is required to integrate our method into the VQSE and VQSD algorithms. Namely, we require linear overhead for VQSE and quadratic overhead for VQSD.

\subsection{Sampling from the dataset}\label{sct_sampling}

While the overhead with the deterministic approach is small, this overhead can be reduced even further via random sampling with $M$ samples. Namely, the scaling with the number of states per cost function evaluation can be improved by sampling from the distribution of states making up the dataset, rather than preparing each state deterministically and classically combining the results. 

The value of the cost function is a scalar quantity. Therefore, we can explore the number of samples necessary to obtain a good estimate of the cost function using Hoeffding's inequality. We first explore how this applies in the case of VQSD. In this case the cost function consists of two terms which both have a non-linear dependence on $\rho$. However, one can bring each term together and consider them to be one observable evaluation. We show that using Hoeffding's inequality allows one to bound the deviation of the observable from its true value with the number of samples and therefore control the overall error in the cost function estimation.

\begin{proposition}
Suppose that we estimate $C_\textrm{VQSD}(U(\boldsymbol{\alpha}), \rho)$ for some unitary $U(\boldsymbol{\alpha})$ by randomly sampling from the dataset $\{\ket{\psi^{(i)}}\}$ that decomposes $\rho$. To ensure that the deviation from the true cost function value $C_\textrm{VQSD}(U(\boldsymbol{\alpha}), \rho)$ is smaller than $\varepsilon$ with probability $(1-\delta)$, it suffices for the number of samples to be
\begin{equation}
M= \frac{9\log(2/\delta)}{2\varepsilon^{2}}
\end{equation}
\end{proposition}
\begin{proof}
As previously stated the two terms in the VQSD cost function can be exactly evaluated using a destructive swap test and a so called DIP test for each term \cite{larose2019variational}. To perform the swap test one measures the expectation value of the swap operator. For the DIP test one evolves the state by a CNOT ladder and measured the projector onto the all zero state on the first system. Therefore, the cost function can be written as
\begin{align}
    &C_\textrm{VQSD}(U(\boldsymbol{\alpha}), \rho) = \Tr(\left(\rho \otimes \rho\right)\text{SWAP}) \nonumber \\  & -  \Tr(\left(\rho \otimes \rho\right)U^{\dagger}(\boldsymbol{\alpha})U^{\dagger}_{D}(\ketbra{0}{0}\otimes \mathbb{1})U_{D}U(\boldsymbol{\alpha})),
\end{align}
where we have used the cyclic property of the trace. We can bring these two terms together and write the cost function as follows:
\begin{align}
    C_\textrm{VQSD}(U(\boldsymbol{\alpha}), \rho) = \Tr(\left(\rho \otimes \rho\right)H),
\end{align}
where $H = \text{SWAP} - U(\boldsymbol{\alpha})U^{\dagger}_{D}(\ketbra{0}{0}\otimes \mathbb{1})U_{D}U(\boldsymbol{\alpha})$.
Using $\rho=\sum\limits_{i}^{N}p^{(i)} \ketbra{\psi^{(i)}}{\psi^{(i)}}$ then leads to the expression,
\begin{align}
    &C_\textrm{VQSD}(U(\boldsymbol{\alpha}), \rho) = \nonumber\\ &\sum\limits_{i,j=1}^{N}p^{(i)}p^{(j)}\Tr\left(\big(\ketbra{\psi^{(i)}}{\psi^{(i)}} \otimes \ketbra{\psi^{(j)}}{\psi^{(j)}}\big)H\right) \ .
\end{align}
Relabeling indices from pairs of $i,j$ to $k$ and rewriting the above equation gives
\begin{align}
    &C_\textrm{VQSD}(U(\boldsymbol{\alpha}), \rho) = \sum\limits_{k = 1}^{N^{2}}p_{k}\Tr\left(\ketbra{\tilde{\psi}^{(k)}}{\tilde{\psi}^{(k)}}H\right).
\end{align}
Therefore, we can consider sampling from the set of states $\tilde{\rho} = \sum_{k=1}^{N^{2}} p^{(k)} \ketbra{\tilde{\psi}^{(k)}}{\tilde{\psi}^{(k)}}$ and estimating the value of the cost function over the sampled states. First consider the quantity 
\begin{equation}
  S_{M} = \frac{1}{M}\sum_{m=1}^{M}X^{(m)} \ ,
\end{equation}
where 
\begin{equation}
    X^{(m)} = \Tr\left(\ketbra{\tilde{\psi}^{(m)}}{\tilde{\psi}^{(m)}}H\right).
\end{equation}
Let us note that every term $X^{(m)} \in [-2,1]$. This follows from the fact that the spectrum of SWAP is in $[-1,1]$ while the spectrum of $\ketbra{0}{0}\otimes \mathbb{1}$ is in $[0,1]$, and hence the overall spectrum of $H$ is in $[-2,1]$. Furthermore, noting that $\mathbb{E}(S_{M}) = C_\textrm{VQSD}(U(\boldsymbol{\alpha}), \rho)$, we can apply Hoeffding's inequality,
\begin{equation}
    \mathcal{P}(|S_{M} - C_\textrm{VQSD}(U(\boldsymbol{\alpha}), \rho)| \geq \varepsilon) \leq 2\exp\left(\frac{-2\varepsilon^{2} M}{9}\right). 
\end{equation}
The probability of the estimate $S_{M}$ deviating from the true value $C_\textrm{VQSD}(U(\boldsymbol{\alpha}), \rho)$ by more than some value $\varepsilon$ decreases exponentially with the number of samples $M$. Defining this probability as $\delta$ leads to an estimate of the lower bound on the number of samples to bound $\varepsilon$ with confidence $1-\delta$,
\begin{equation}
    M = \frac{9\log(2/\delta)}{2\varepsilon^{2}}.
\end{equation}
\end{proof}
A similar approach can be taken when considering VQSE. In that case the cost function only consists of the evaluation of one observable and one copy of the state $\rho$. This makes the analysis simpler, but the conclusion is the same. Furthermore, $|C_\textrm{VQSE}| \leq \| H \|_{\infty}$ with $H$ given by \eqref{eqn_Hvqse}, and this Hamiltonian norm ends up appearing in the exponent in Hoeffding's inequality. This leads to a similar proposition for the case of VQSE:
\begin{proposition}
Suppose that we estimate $C_\textrm{VQSE}(U(\boldsymbol{\alpha}), \rho)$ for some unitary $U(\boldsymbol{\alpha})$ by randomly sampling from the dataset $\{\ket{\psi^{(i)}}\}$ that decomposes $\rho$. To ensure that the deviation from the true cost function value $C_{VQSE}(U(\boldsymbol{\alpha}), \rho)$ is smaller than $\varepsilon$ with probability $(1-\delta)$, it suffices for the number of samples to be
\begin{equation}
M= \frac{2 \| H \|^{2}_{\infty}\log(2/\delta)}{\varepsilon^{2}}
\end{equation}
with $H$ given by \eqref{eqn_Hvqse}.
\end{proposition}

One key feature of this result is the absence of $N$, the total number of states in the dataset. Therefore, we can conclude that the number of samples necessary in order to obtain a good cost function estimate does not scale with the number of datapoints that make up the dataset. Sampling to obtain good cost function estimates then in turn would lead to a successful optimization, and therefore accurate determination of the eigenvalues and eigenvectors corresponding to the principal components. 

The above analysis can also be extended to the number of shots needed to obtain an accurate cost function value. If one considers the limiting case where only one shot is used for the evaluation of the cost function for each state the same argument given above applies. This leads to the conclusion that only $\frac{9\log(2/\delta)}{2\varepsilon^{2}}$ shots are necessary to compute the cost function with accuracy $\varepsilon$ and confidence $\delta$ for the case of $C_\textrm{VQSD}$. Each shot is taken while also randomly sampling from the dataset $\{\ket{\psi^{(i)}}\}$. A similar result follows for $C_\textrm{VQSE}$.

Overall, when looking at using our method to prepare the covariance matrix as a subroutine in VQSD and VQSE we find the sampling overhead necessary to still obtain an optimization close to the exact case is favorable. Indeed, we find essentially no dependence on the total size of the dataset.

\subsection{Remark about original quantum PCA algorithm}\label{sct_originalqPCA}

We emphasize that we focus on VQSD and VQSE instead of the original quantum PCA algorithm~\cite{lloyd2014quantum} because the latter is a long-term algorithm and we are aiming at a near-term approach to quantum PCA. The simplicity of our method is ideally suited to NISQ (Noisy Intermediate Scale Quantum) devices. Hence our method partners well with other NISQ algorithms (for state diagonalization).

Nevertheless, one can apply our approach (of preparing the ensemble average density matrix $\overline{\rho}$) in the context of the original quantum PCA algorithm as well. Recall that their algorithm uses $c$ copies of $\overline{\rho}$ in order to approximately implement the $e^{-i \overline{\rho} t}$ gate in the context of quantum phase estimation (QPE). Here, one also feeds $\overline{\rho}$ as the input state into the QPE circuit. Their algorithm requires $c\in \mathcal{O}(1/\epsilon^3)$ copies of $\overline{\rho}$ in order to extract its eigenvalues and eigenvectors with accuracy~$\epsilon$.

One can employ a deterministic approach to prepare $\overline{\rho}$, as discussed in Sec.~\ref{sct_deterministic}. Operationally speaking, this approach requires the simultaneous decomposition of both the input state $\overline{\rho}$ and the copies of $\overline{\rho}$ that are used to approximate the $e^{-i \overline{\rho} t}$ gate implemented in the QPE circuit. Let us now discuss how this decomposition leads to the expected (i.e., correct) output, due to the linearity of the operations involved. In this setting, one would: (1) choose a state from the dataset (i.e., from the decomposition of $\overline{\rho}$), (2) feed this state into the QPE circuit, (3) estimate observable expectation values on the resulting state (to characterize the eigenvalues and eigenvectors), and (4) average these observable expectation values over all of the states in the dataset. By the linearity of the QPE unitary and of the trace, the average of the expectation values is the same as the expectation value of the average, and hence this method reproduces the correct results. In addition, as previously mentioned one would also need to choose states from the dataset for each ancilla system that is used to approximate the $e^{-i \overline{\rho} t}$ gate using exponential swap gates followed by a partial trace over the ancilla. Once again, due to the linearity of the exponential swap gates and of the partial trace, averaging over the dataset gives the correct result.

In this deterministic setting, one would need $N$ state preparations for each copy of $\overline{\rho}$, i.e., for the input system to QPE and for the ancilla systems using to approximate the $e^{-i \overline{\rho} t}$ gate. If a total of $c$ copies of $\overline{\rho}$ are employed, then one would potentially need $N^c$ state preparations with this approach.

We remark that a random sampling approach, instead of a deterministic approach,  to preparing each copy of $\overline{\rho}$  could alternatively be used. See Sec.~\ref{sct_sampling} for further discussion of this approach. Although a detailed analysis of random sampling in this context is beyond the scope of this work, we do believe that random sampling would likely significantly reduce the number of state preparations required, as compared to the deterministic case. This intuition arises from Hoeffding's bounds, which implies that an estimator will concentrate about its mean value as one increases the number of samples, and this concentration guarantee is essentially independent of the dataset size $N$.

\section{Numerical implementations}\label{sct_numerics}

Here we perform PCA (i.e., diagonalizing $Q$) and simulate quantum PCA (i.e., diagonalizing $\overline{\rho}$) on two different datasets. The first dataset is classical and is the famous MNIST dataset of handwritten digits. The second dataset is quantum dataset of molecular ground states for various interatomic distances.

\subsection{MNIST implementation}

\begin{figure}[t]
    \centering
    \includegraphics[width =\columnwidth]{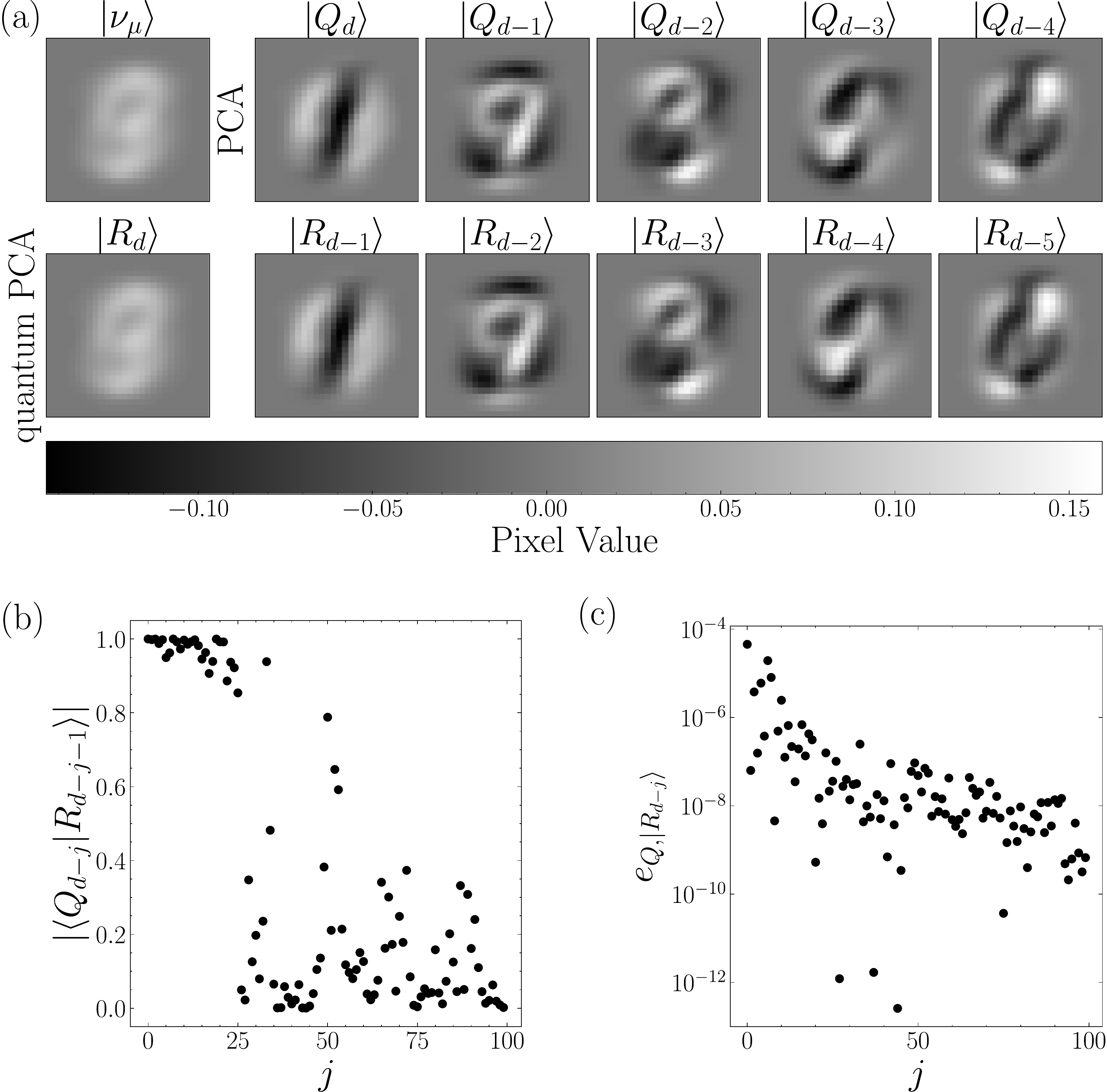}
    \caption{\textbf{Eigenvectors for the MNIST dataset of handwritten digits.} We perform standard PCA (diagonalizing $Q$) and quantum PCA (diagonalizing $\overline{\rho}$). (a) The images of the first $5$ principal components for PCA and the first $6$ for quantum PCA are shown.  The first principal component for quantum PCA is close to (has a large overlap with) the mean vector $\ket{v_{\vec{\mu}}}$, which is also shown. Also, one can visually see that $\ket{Q_{d-j}}$ is quite similar to $\ket{R_{d-j-1}}$ for these eigenvectors. (b) We plot the magnitude of the overlap of the principal components calculated from PCA and quantum PCA, namely, $\ket{Q_{d-j}}$ and $\ket{R_{d-j-1}}$. This overlap remains large when $j$ is less than 25. (c) Finally, we show the eigenvector error, which quantifies how far the eigenvectors of $\overline{\rho}$ are from being eigenvectors of $Q$.}
    \label{fig:MNIST_eigenvectors}
\end{figure}

\begin{figure}[t]
    \centering
    \includegraphics[width =\columnwidth]{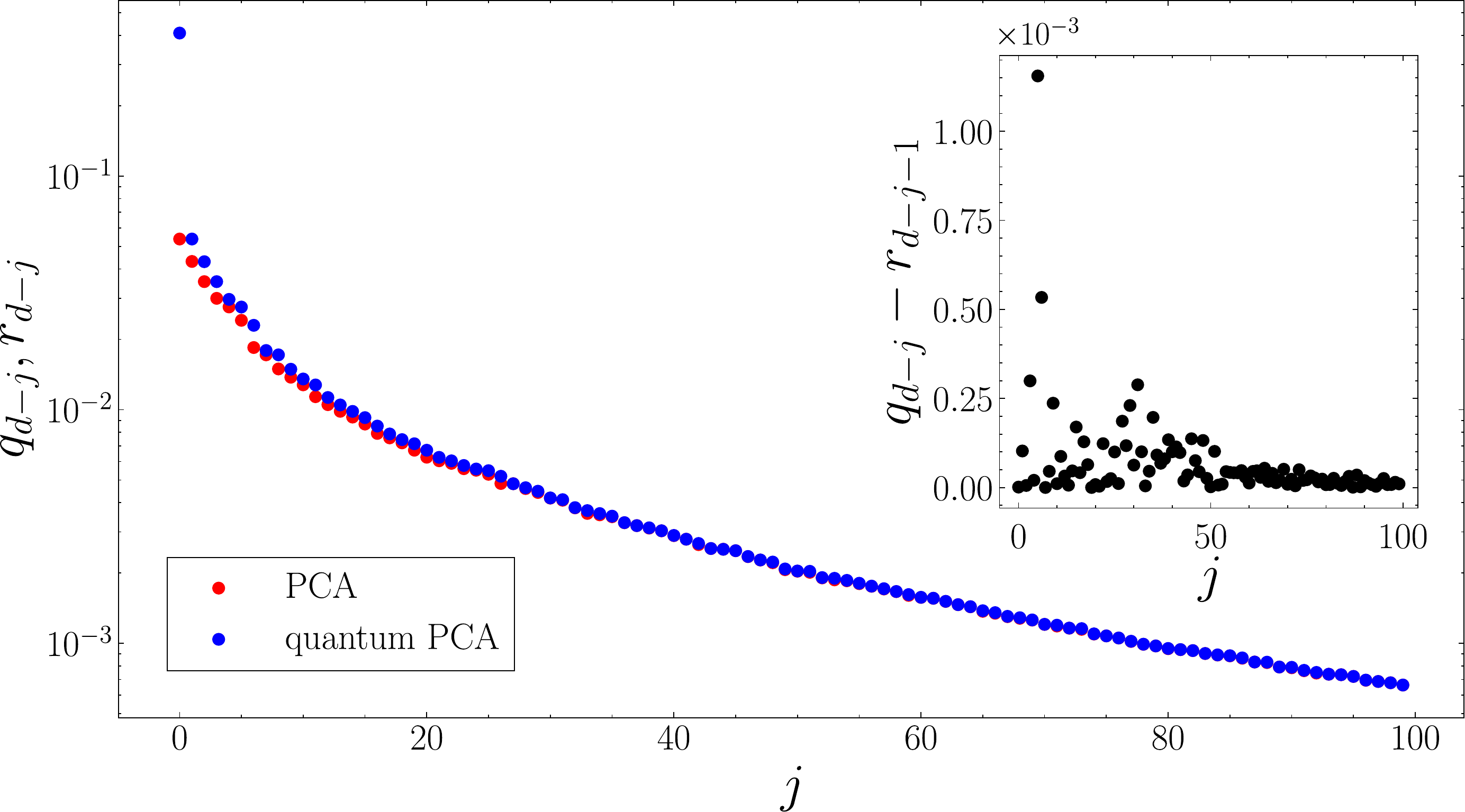}
    \caption{\textbf{Eigenvalues for the MNIST dataset of handwritten digits.} We perform standard PCA (diagonalizing $Q$) and quantum PCA (diagonalizing $\overline{\rho}$). The eigenvalues calculated using PCA and quantum PCA are plotted. Note that $q_{d-j}$ is quite similar to $r_{d-j-1}$ for these principal components. In the inset we plot the difference between the shifted eigenvalues $q_{d-j}$ and $r_{d-j-1}$ calculated using PCA and quantum PCA, respectively.}
    \label{fig:MNIST_eigenvalues}
\end{figure}

From the MNIST dataset, we randomly select $5000$ instances of each digit from $0$ to $9$. We vectorize the $28 \times 28$ grey-scaled images leading to vectors $\vec{y}^{(i)}$. This gives $d=784$ features for standard PCA, and for our simulation of quantum PCA we trivially embed the data in a $d=2^{10} = 1024$ dimensional feature space corresponding to the Hilbert space of 10 qubits. Each vector is normalized such that $(\vec{y}^{(i)})^{T}\vec{y}^{(i)}=1$ This results in a data set $D = \{\vec{y}^{(i)}\}_{i=1}^{N}$ with $N=50000$. We assume a uniform probability distribution so that the ensemble is $E = \{1/N,\vec{y}^{(i)} \}$.  

In our numerics, we perform PCA and we simulate quantum PCA for this dataset. We then compare the principal components produced by each approach. For standard PCA, we diagonalize the covariance matrix $Q$ formed from the vectors in $D$. The top $n$ principal components are the eigenvectors corresponding to the $n$ largest eigenvalues. For our simulation of quantum PCA, we prepare $\bar{\rho}$ as outlined in the text above. We then diagonalize $\bar{\rho}$ to give the quantum principal components.

\begin{figure}[t]
    \centering
    \includegraphics[width =1\columnwidth]{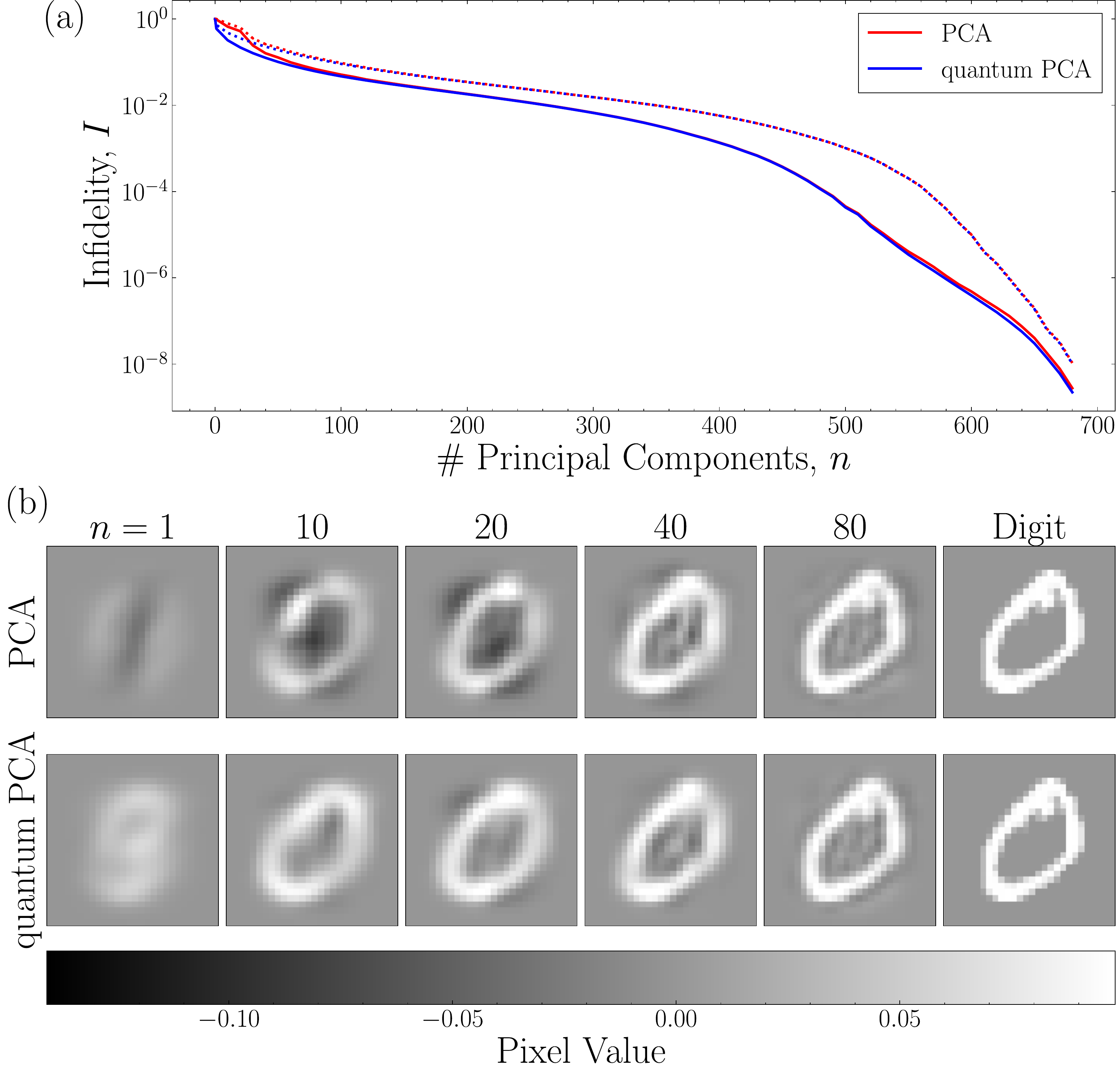}
    \caption{\textbf{Projection onto $n$ principal components for PCA and quantum PCA of the MNIST dataset.} (a) The plot shows the median (solid) and upper 90\% interval (dotted) infidelity ($I$) values between the projected state constructed using $n$ principal components calculated from PCA and quantum PCA. The median and upper interval are calculated over all $50000$ digits. (b) We show how the images are reproduced when using $n =1 $ to $n = 80$ principal components from PCA and quantum PCA in the upper and lower rows respectively. }
    \label{fig:infidelity}
\end{figure}

For verification purposes, we also performed standard PCA on a symmetrized dataset
\begin{equation}
    \tilde{D} = D\cup (-D) = \{\vec{y}^{(i)}\}_{i=1}^{N} \cup \{ - \vec{y}^{(i)}\}_{i=1}^{N}\,.
\end{equation}
We verified that we obtain the same spectrum from performing PCA with this symmetrized dataset as we obtain with our simulation of quantum PCA. This is expected as both scenarios effectively involve a symmetrized dataset.

\subsubsection{Eigenvalues and Eigenvectors}

We now discuss the spectrum that we obtained in the two cases. In Fig.~\ref{fig:MNIST_eigenvectors}(a) we show the eigenvectors associated with the first $5$ principal components for PCA and the first $6$ for quantum PCA. We see an intriguing correspondence between the eigenvectors obtained via standard and quantum PCA. Specifically, $\ket{Q_{d-j}}$ appears to be very similar to $\ket{R_{d-j-1}}$. Indeed, one can see in Fig.~\ref{fig:MNIST_eigenvectors}(b) that the overlap $|\ip{Q_{d-j}}{R_{d-j-1}}|^{2}$ is close to one for $j< 25$. For larger $j$ this overlap is smaller, however this does not mean that the $\ket{R_{d-j-1}}$ states are not approximate eigenvectors of $Q$. To clarify this point, we also plot the eigenvector error in Fig.~\ref{fig:MNIST_eigenvectors}(c). For all values of $j$ shown, the eigenvector error is below $10^{-4}$. This suggests that the eigenvectors of $\overline{\rho}$ serve as approximate eigenvectors for $Q$, for this dataset.

We see a similar pattern in the eigenvalues, shown in Fig.~\ref{fig:MNIST_eigenvalues}. Excluding the largest eigenvalue from quantum PCA, the eigenvalues match very well for the two methods. More specifically, $r_{d-j-1}$ is very close to $q_{d-j}$, and this is shown in more detail in the inset of Fig.~\ref{fig:MNIST_eigenvalues}.

The close correspondence between both the eigenvectors and the eigenvalues for the two methods is likely a consequence of Prop.~\ref{prop_spectrum_evectors}. The assumption in Prop.~\ref{prop_spectrum_evectors} is almost satisfied, i.e., one of the eigenvectors of $\overline{\rho}$ is close to the mean vector. Specifically, the first
principal component of $\overline{\rho}$, which is $\ket{R_d}$ and is displayed in Fig.~\ref{fig:MNIST_eigenvectors}(a), has an infidelity of roughly $2.8 \times 10^{-4}$ with the mean vector $\ket{v_{\vec{\mu}}}$. The $\ket{R_d}$ eigenvector appears to capture the bias of the dataset, i.e., the fact that the dataset is biased towards white colored images. The fact that $\ket{R_d}$ is close to the mean vector suggests that the spectral decomposition in \eqref{eqn_spectrum_inheritance} is almost valid. Hence the spectra of $\overline{\rho}$ and $Q$ are close to matching.

\subsubsection{Principal component projections}

In order to assess how the two methods for calculating the principal components perform, we construct a projected image by projecting each image onto a reduced subspace
\begin{equation}
\ket{\hat{y}_{n}^{(i)}} = \sum_{j=0}^{n-1} \braket{\chi_{d-j}}{y^{(i)}} \ket{\chi_{d-j}} ,
\end{equation}
where $\ket{\hat{y}_{n}^{(i)}}$ is the projected vector using $n$ principal components $\ket{\chi_{d-j}}$ calculated from either PCA or quantum PCA. We can explore the infidelity between the projected vector and the original as 
\begin{equation}\label{eqn_infidelity}
I = 1 - |\braket{\hat{y}_{N_{c}}^{(i)}}{y^{(i)}}|. 
\end{equation}
In Fig.~\ref{fig:infidelity}(a) we show the median and 90\% interval of the infidelity over every image as a function of $n$ for the classical and quantum cases. The performance of PCA and quantum PCA is similar with roughly $100$ components necessary for a median infidelity of $0.1$. Therefore, despite the differences in the calculated eigenvalues and eigenvectors, both methods can be used to accurately compress the data on average.

In Fig.~\ref{fig:infidelity}(b) we show how one projected image in the MNIST dataset appears visually as we increase the number of principal components used. There is a noticeable difference in the $n=1$ case, due to $\ket{R_d}$ being quite different from $\ket{Q_d}$. However, by $n = 80$ the two projected images appear visually similar and also close to the true image. It therefore appears that quantum PCA is a successful surrogate for PCA, for this dataset.

\subsection{Molecular ground state implementation}

We simulate the task of performing quantum PCA on molecular ground states of the H${}_{2}$ molecule in the $6$-$31$g basis and the BeH${}_{2}$ in the sto-$3$g basis. These implementations require system sizes of 8 qubits and 14 qubits, respectively, corresponding to feature space dimensions of $d = 256$ and $d = 16384$. The ground state is calculated by classically simulating the variational quantum eigensolver (VQE)~\cite{peruzzo2014variational} for $401$ equally spaced interatomic distances $r \in [0.3, 2.3]$. 

We find that the principal components and eigenvalues calculated using standard PCA and quantum PCA are almost identical in this case. This is due to the fact that our simulations of VQE produce ground states with random global phases, and this naturally results in a dataset that is approximately centered, especially as one increases the number of datapoints. Hence, in this case, Prop.~\ref{prop_centered} applies, and diagonalizing $\overline{\rho}$ is equivalent to diagonalizing $Q$. (This point would be irrelevant if we assume that the dataset is truly generated on a quantum device, rather than through our classical simulations, since global phase is unphysical and then Prop.~\ref{prop_quantumdataset} would apply.)

Recall from Eq.~\eqref{eqn_infidelity} that the infidelity quantifies the inability to recover the original data from the compressed data. In Fig.~\ref{fig:H2_BeH2}(a) we show the median and 90\% interval of the infidelity over every ground state as a function of $n$ for the two molecules. The projected ground states have a very low infidelity with the actual ground states for low numbers of principal components, especially in the case of H${}_2$. Therefore, there exists an accurate, more efficient representation of these states, which can be calculated with quantum PCA. In Fig.~\ref{fig:H2_BeH2}(b) we show the infidelity $I$ as a function of the interatomic distance $r$ for projected states calculated with $n=1,2,3,4$ principal components for both molecules. One can see that the curves monotonically decrease with $n$.

\begin{figure}
    \centering
    \includegraphics[width =1\columnwidth]{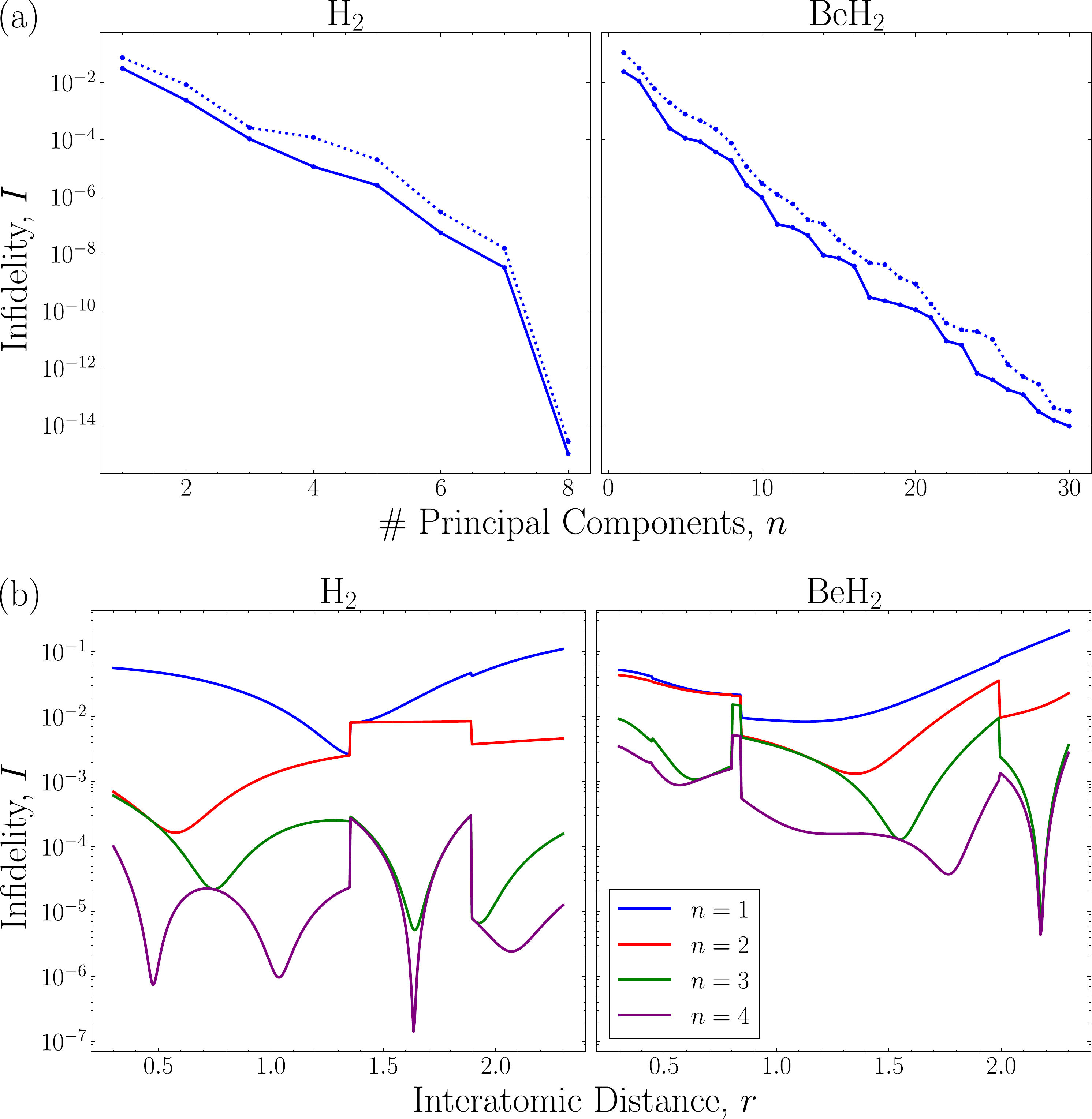}
    \caption{\textbf{Quantum PCA performed on ground states of molecules at different interatomic distances.} (a) The plot shows the infidelity ($I$) between the original state and the projected state constructed using $n$ principal components calculated by simulating quantum PCA. The median (solid line) and upper 90\% interval (dotted line)  are calculated over all $401$ interatomic distances $r$. (b) The infidelity is plotted as a function of interatomic distance $r$ for projected states calculated with $n=1,2,3,4$ principal components. }
    \label{fig:H2_BeH2}
\end{figure}

\section{Discussion}

Recently there is renewed interest in quantum algorithms for PCA. The history here is quite interesting. Lloyd et al.'s quantum PCA algorithm was proposed in 2014 and highlighted a potential exponential speedup over classical algorithms~\cite{lloyd2014quantum}, and this also led to some near-term proposals for quantum PCA~\cite{larose2019variational,cerezo2020variational,verdon2019quantum,ezzell2022quantum}. However, a spooky paper posted on Halloween of 2018 by Tang presented a ``dequantized'' classical algorithm that achieved the same asymptotic scaling as quantum PCA~\cite{tang2021quantum}, suggesting that quantum speedup for quantum PCA was a false promise. But this was not the end of the story. In December 2021, a team from Caltech and Google posted two papers~\cite{cotler2021revisiting,huang2021quantumadvantage} arguing that exponential quantum speedup is possible for quantum PCA, and even for near-term algorithms, since (as they argue)  dequantized classical algorithms are artificially given too much power via their mode of access to quantum state amplitudes. As a consequence of this latest work, we are left with the exciting possibility that exponential quantum speedup remains possible for quantum PCA, particularly for analysis of quantum data. It is worth remarking that the Tang's dequantization results still apply to analysis of classical data~\cite{cotler2021revisiting}, although this does not preclude the possibility of modest (e.g., constant factor) quantum speedups for classical data analysis~\cite{arrazola2019quantum}.

This interesting historical trajectory for quantum PCA makes our work even more important. A crucial piece of the puzzle was missing in this field, and that was a method for preparing the covariance matrix, given a dataset of quantum states. We have made major strides towards filling this gap, by proposing the ensemble average density matrix as a surrogate for the covariance matrix. We argued that this surrogate was equal to the covariance matrix for arbitrary quantum datasets or for centered classical datasets. Given the potential for exponential speedup with quantum PCA for quantum datasets~\cite{cotler2021revisiting,huang2021quantumadvantage}, our results for these datasets are technologically important. Therefore we placed significant emphasize on PCA for quantum datasets in this work, providing a detailed discussion of complex random variables (see Sec.~\ref{sct_background}) and showcasing an implementation for molecular ground states (see Sec.~\ref{sct_numerics}). 

On the other hand, for uncentered classical datasets, we showed that diagonalizing this surrogate matrix corresponds to ``PCA without centering'', or PCA for a symmetrized version of the dataset as shown in Fig.~\ref{fig:symmetrization}. We nevertheless derived results that bound the deviation of the spectrum obtained with our method from the true PCA spectrum, indicating a close correspondence with true PCA. (At the mathematical level, our results here generalize and extend those of Ref.~\cite{cadima2009relationships}.) Hence, we believe that our method will be useful even for uncentered classical datasets. 

We remark that a recent experimental quantum PCA implementation~\cite{xin_experimental_2021} employed a non-scalable method for covariance matrix preparation. They classically optimized over circuits to prepare the covariance matrix, which has exponential scaling with problem size. Moreover, that approach would be unnatural for quantum datasets, since one would have to first readout the quantum states, leading to an additional source of exponential scaling. Thus, our approach of preparing $\overline{\rho}$ fills an important gap in the literature, towards scalable covariance matrix preparation, especially for quantum datasets~\footnote{For classical datasets, the scaling of our method will be determined by the scaling of the amplitude encoding step, which is still an active area of research~\cite{grover2000synthesis,grover2002creating,plesch2011quantum,schuld2018supervised,sanders2019black,nakaji2021approximate,marin2021quantum,zoufal2019quantum}}.

Natural future work would be to actually implement our method on real quantum hardware. Combining our method for preparing the covariance matrix with other near-term methods for extracting the spectrum~\cite{larose2019variational,cerezo2020variational,verdon2019quantum,ezzell2022quantum} would lead to a near-term approach for quantum PCA. Indeed the nice feature of our method is how simple and easy-to-implement it is on near-term quantum hardware.

\begin{acknowledgments}

We thank Kunal Sharma for helpful and insightful discussions. MHG and PJC were supported by the U.S. Department of Energy (DOE), Office of Science, Office of Advanced Scientific Computing Research, under the Quantum Computing Application Teams (QCAT) program. MC was supported by the Laboratory Directed Research and Development (LDRD) program of Los Alamos National Laboratory (LANL) under project number 20210116DR. LC was supported by the LDRD program of LANL under project number 20200056DR. MC and PJC were also initially supported by the LANL ASC Beyond Moore's Law project.

\end{acknowledgments}

\bibliography{quantum.bib}

\appendix

\section{Weyl's theorem for eigenvalues}\label{Appendix_Weyl}

For our purposes, we will need a theorem from Weyl about the eigenvalues of Hermitian matrices. This theorem is sometimes called Weyl's inequality. A detailed discussion and proof of this theorem can be found in the textbook of Horn and Johnson~\cite{horn2012matrix}. We repeat the theorem here, as it is stated in Ref.~\cite{horn2012matrix}.

\begin{lemma}
Let $A$ and $B$ be Hermitian $d\times d$ matrices. Let the respective eigenvalues of $A$, $B$, and $A+B$ be $\{\lambda_j(A)\}_{j=1}^d$, $\{\lambda_j(B)\}_{j=1}^d$, and $\{\lambda_j(A+B)\}_{j=1}^d$. Suppose that the eigenvalue in these sets are listed in non-decreasing order. Then
\begin{equation}\label{eqn_Weyl1}
    \lambda_j(A+B) \leq \lambda_{j+k}(A) + \lambda_{d-k}(B)
\end{equation}
for each $j = 1, ..., d$ and $k = 0, 1,..., d-j$. In addition, 
\begin{equation}\label{eqn_Weyl2}
    \lambda_{j-k+1}(A)+\lambda_{k}(B) \leq \lambda_{j}(A+B)
\end{equation}
for each $j = 1, ..., d$ and $k=1,..., j$.
\end{lemma}
By setting $k =0$ and $k=1$ in Eq.~\eqref{eqn_Weyl1} and by setting $k=1$ in Eq.~\eqref{eqn_Weyl2}, we obtain the following corollary.
\begin{lemma}
Let $A$ and $B$ be Hermitian $d\times d$ matrices. Let the respective eigenvalues of $A$, $B$, and $A+B$ be $\{\lambda_j(A)\}_{j=1}^d$, $\{\lambda_j(B)\}_{j=1}^d$, and $\{\lambda_j(A+B)\}_{j=1}^d$. Suppose that the eigenvalue in these sets are listed in non-decreasing order. Then
\begin{align}
        &\lambda_j(A+B) \leq \lambda_{j}(A) + \lambda_{d}(B)\quad\text{for }j = 1, ..., d\\
        &\lambda_j(A+B) \leq \lambda_{j+1}(A) + \lambda_{d-1}(B)\quad\text{for }j = 1, ..., d-1\\
        &\lambda_j(A)+\lambda_1(B) \leq \lambda_{j}(A+B)\quad\text{for }j = 1, ..., d\,.
\end{align}
\end{lemma}
Finally, by specializing the result even further, to the case where $B$ is rank-one with a positive eigenvalue, we obtain the following result.
\begin{lemma}
Let $A$ and $B$ be Hermitian $d\times d$ matrices. Let the respective eigenvalues of $A$, $B$, and $A+B$ be $\{\lambda_j(A)\}_{j=1}^d$, $\{\lambda_j(B)\}_{j=1}^d$, and $\{\lambda_j(A+B)\}_{j=1}^d$. Suppose that the eigenvalue in these sets are listed in non-decreasing order. Suppose that $B$ is rank-one and that the only non-zero eigenvalue of $B$ is $\lambda_B >0$. Then
\begin{align}
        \lambda_j(A+B) &\leq \lambda_{j}(A) + \lambda_B\quad\text{for }j = 1, ..., d\\
        \lambda_j(A+B) &\leq \lambda_{j+1}(A)\quad\text{for }j = 1, ..., d-1 \\
        \lambda_j(A) &\leq \lambda_{j}(A+B)\quad\text{for }j = 1, ..., d\,.
\end{align}
\end{lemma}

\section{Relating the diagonal elements of $\overline{\rho}$ and $Q$}\label{Appendix_diagonalelements}

In the main text, we related the eigenvalues and eigenvectors of $\overline{\rho}$ and $Q$. In this appendix, we note that one can also bound the deviation of the diagonal elements of $\overline{\rho}$ and $Q$. This is different from bounding the deviation of the ordered eigenvalues, and hence the following proposition is distinct from Prop.~\ref{prop_spectrum_evalues}. In fact, there is a simple equation that relates the diagonal elements of $\overline{\rho}$ and $Q$, as follows.
\begin{proposition}\label{prop_diagonalelements}
Consider a classical dataset of pure states $D_{\ket{\psi}} = \{\ket{\psi^{(i)}}\}_{i=1}^N$, with the ensemble denoted as $E_{\ket{\psi}} = \{p^{(i)},\ket{\psi^{(i)}}\}_{i=1}^N$. Let $\ket{\phi}$ be any quantum state and let $r_{\ket{\phi}} = \mte{\phi}{\overline{\rho}}$ and $q_{\ket{\phi}} = \mte{\phi}{Q}$. Then
\begin{equation}
    q_{\ket{\phi}} \leq r_{\ket{\phi}} \leq q_{\ket{\phi}} + \|\vec{\mu}\|^2\,.
\end{equation}
More specifically, the following equation holds:
\begin{equation}
    r_{\ket{\phi}}  = q_{\ket{\phi}}+ \|\vec{\mu}\|^2  |\ip{v_{\vec{\mu}}}{\phi}|^2
\end{equation}
where $\vec{\mu}$ is the mean vector and $\ket{v_{\vec{\mu}}} = \vec{\mu} / \|\vec{\mu}\|$ is the normalized mean vector.
\end{proposition}
\begin{proof}
The inequality follows from the equation in the proposition. The equation follows from $\overline{\rho} = Q+M$ and $M = \|\vec{\mu}\|^2 \dya{v_{\vec{\mu}}}$, which gives the desired result through:
\begin{equation}
  \mte{\phi}{\overline{\rho}} = \mte{\phi}{Q} + \mte{\phi}{M}\,.
\end{equation}
\end{proof}

\section{Extension to mixed state datasets}\label{Appendix_mixedstates}

We now discuss how the above results can be extended to mixed state datasets. For mixed state datasets, there exist straightforward mathematical generalizations of the previous results. The conceptual interpretation of such results is less obvious. Nevertheless we will discuss this below.

\subsection{Mixed-state dataset}

Consider a dataset of mixed states, which we denote as
\begin{equation}\label{eqn_dataset_mixed1_app}
    D_{\rho} = \{\rho^{(i)}\}_{i=1}^N\,.
\end{equation}
We assume a probability distribution over datapoints given by $P = \{p^{(i)}\}$.  Hence, one can defined the corresponding ensemble for this dataset as:
\begin{equation}
 E_{\rho}= \{p^{(i)},\rho^{(i)} \}_{i=1}^N   
\end{equation}
For this ensemble, the ensemble average density matrix is:
\begin{align}
    \overline{\rho} &= \sum_{i=1}^N p^{(i)}\rho^{(i)}\,.
\end{align}

\subsection{Effective pure-state dataset}

Let us note that each mixed-state datapoint can be decomposed as a convex combination of pure states: 
\begin{equation}
    \rho^{(i)} = \sum_{m=1}^{R_i} s^{(i,m)} \dya{\psi^{(i,m)}}
\end{equation}
Hence, we can rewrite the ensemble average density matrix as:
\begin{align}
    \overline{\rho} &= \sum_{i=1}^N p^{(i)}\rho^{(i)} \\
    &= \sum_{i=1}^N \sum_{m=1}^{R_i} p^{(i)} s^{(i,m)}\dya{\psi^{(i,m)}}\label{eqn_rhobar_appendix}
\end{align}
We can introduce an index $k = (i,m)$ and then we have
\begin{align}
    \overline{\rho} = \sum_{k=1}^{\widehat{N}}  \widehat{p}^{(k)} \dya{\psi^{(k)}}
\end{align}
where 
\begin{equation}\label{eqn_effective_parameters}
\widehat{N} = \sum_{i=1}^N R_i\quad\text{and}\quad \widehat{p}^{(k)} = p^{(i)} s^{(i,m)}\,.
\end{equation}
The above equation suggests that we could interpret
\begin{equation}
    \widehat{D}_{\rho} = \{ \dya{\psi^{(k)}}  \}_{k=1}^{\widehat{N}}
\end{equation}
as an effective dataset, with pure state datapoints. The corresponding effective ensemble is then 
\begin{equation}\label{eqn_effective_ensemble}
    \widehat{E}_{\rho} = \{ \widehat{p}^{(k)}, \dya{\psi^{(k)}}  \}_{k=1}^{\widehat{N}}
\end{equation}

With these definitions in hand, we can now see how our main results, for pure state datasets, can generalize to mixed state datasets. If we are willing to reinterpret the dataset as $ \widehat{D}_{\rho}$ and the ensemble as $ \widehat{E}_{\rho}$, then all of our results can be extended to this case. 
The idea is that we will relate $\overline{\rho}$ to the covariance matrix for an unraveling of the ensemble $ \widehat{E}_{\rho}$. Consequently, the conceptual interpretation of diagonalizing $\overline{\rho}$ is that it corresponds to performing PCA on a dataset composed of the pure states that decompose the mixed states in the original dataset. This interpretation is made precise in Prop.~\ref{prop_mixedstates_appendix} below.

\subsection{Theoretical results}

Mixed-state datapoints naturally lack a global phase, since global phases are only relevant for statevector datasets. Therefore, in the context of mixed-state datasets, it seems natural to consider our results above for the case where global phase information is not important, or not relevant. Such results were given in Lemma~\ref{lemma2} and Prop.~\ref{prop_quantumdataset}. Hence, in what follows we will state analogs of Lemma~\ref{lemma2} and Prop.~\ref{prop_quantumdataset} for mixed-state datasets.

\begin{lemma}\label{lemma_mixed_appendix}
Consider an ensemble of mixed states $E_{\rho} = \{p^{(i)},\rho^{(i)}\}_{i=1}^N$. Let $\widehat{E}_{\rho}=\{ \widehat{p}^{(k)}, \dya{\psi^{(k)}}  \}_{k=1}^{\widehat{N}}$ be the corresponding effective ensemble of pure states, as given in \eqref{eqn_effective_parameters} and \eqref{eqn_effective_ensemble}. Then, there always exists a statevector ensemble $E_{\ket{\tilde{\psi}}} = \{\tilde{p}^{(j)},\ket{\tilde{\psi}^{(j)}}\}_{j=1}^{\tilde{N}}$ that satisfies the following conditions:
\begin{itemize}
    \item $E_{\ket{\tilde{\psi}}}$ physically corresponds to the aforementioned $\widehat{E}_{\rho}$, in the sense that applying the outer product mapping leads to $\PC(E_{\ket{\tilde{\psi}}}) = \widehat{E}_{\rho}$. In other words, $E_{\ket{\tilde{\psi}}}$ is an unraveling of $\widehat{E}_{\rho}$.
    \item $E_{\ket{\tilde{\psi}}}$ is centered, i.e., the mean value of all features is zero.
\end{itemize}
\end{lemma}
\begin{proof}
The proof is essentially the same as the proof of Lemma~\ref{lemma2}. In fact, one can view this as a corollary of Lemma~\ref{lemma2}, where one applies Lemma~\ref{lemma2} to the ensemble $\widehat{E}_{\rho}$.

A statevector ensemble that satisfies the two criteria stated in the lemma is:
\begin{align}\label{eqn_centeredensemble_appendix}
    E_{\ket{\tilde{\psi}}}  =& \bigg\{ (\frac{\widehat{p}^{(1)}}{2},\ket{\psi^{(1)}}),..., (\frac{\widehat{p}^{(\widehat{N})}}{2},\ket{\psi^{(\widehat{N})}}),\notag\\
    &(\frac{\widehat{p}^{(1)}}{2},-\ket{\psi^{(1)}}),...,(\frac{\widehat{p}^{(\widehat{N})}}{2},-\ket{\psi^{(\widehat{N})}})\bigg\}\,.
\end{align}
Note that this ensemble is centered, due to its symmetric nature.

In addition, one can see that $\PC(E_{\ket{\tilde{\psi}}}) = \widehat{E}_{\rho}$, and hence $E_{\ket{\tilde{\psi}}}$ is an unraveling of $\widehat{E}_{\rho}$. This follows because $\PC(E_{\ket{\tilde{\psi}}})$ has $\widehat{N}$ datapoints that appear twice in ensemble, and these redundant datapoints can be aggregated to give $E_{\ket{\tilde{\psi}}}$.
\end{proof}

With the previous lemma in hand, we can state the following proposition, which is our main result for mixed-state datasets. Note that the following proposition generalizes the result in Prop.~\ref{prop_quantumdataset}.
\begin{proposition}\label{prop_mixedstates_appendix}
Consider an ensemble of mixed states $E_{\rho} = \{p^{(i)},\rho^{(i)}\}_{i=1}^N$. Let $\widehat{E}_{\rho}=\{ \widehat{p}^{(k)}, \dya{\psi^{(k)}}  \}_{k=1}^{\widehat{N}}$ be the corresponding effective ensemble of pure states, as given in \eqref{eqn_effective_parameters} and \eqref{eqn_effective_ensemble}. Then, there always exists a statevector ensemble $E_{\ket{\tilde{\psi}}} = \{\tilde{p}^{(j)},\ket{\tilde{\psi}^{(j)}}\}_{j=1}^{\tilde{N}}$ that satisfies the following conditions:
\begin{itemize}
    \item $E_{\ket{\tilde{\psi}}}$ physically corresponds to the aforementioned $\widehat{E}_{\rho}$, in the sense that applying the outer product mapping leads to $\PC(E_{\ket{\tilde{\psi}}}) = \widehat{E}_{\rho}$. In other words, $E_{\ket{\tilde{\psi}}}$ is an unraveling of $\widehat{E}_{\rho}$.
    \item The covariance matrix $Q(E_{\ket{\tilde{\psi}}})$ for $E_{\ket{\tilde{\psi}}}$ is equal to the ensemble average density matrix for $E_{\rho}$:
    \begin{equation}
    Q(E_{\ket{\tilde{\psi}}}) = \overline{\rho}(E_{\rho})\,.
\end{equation}
\end{itemize}
\end{proposition}
\begin{proof}
We will prove this by constructing an $E_{\ket{\tilde{\psi}}}$ that satisfies the required conditions. Specifically we will use the ensemble in \eqref{eqn_centeredensemble_appendix}. Of course, we already showed above that this ensemble is an unraveling of $\widehat{E}_{\rho}$. 
So we just need to prove the second condition. Let us first note that $E_{\ket{\tilde{\psi}}}$ is centered, as discussed above. Hence we can apply Prop.~\ref{prop_centered} to see that
\begin{equation}
    Q(E_{\ket{\tilde{\psi}}}) = \overline{\rho}(E_{\ket{\tilde{\psi}}})\,.
\end{equation}
Next we note that $\overline{\rho}(E_{\ket{\tilde{\psi}}}) = \overline{\rho}(\widehat{E}_{\rho})$, since $E_{\ket{\tilde{\psi}}}$ is an unraveling of $\widehat{E}_{\rho}$. Finally, it is clear that $\overline{\rho}(\widehat{E}_{\rho}) = \overline{\rho}(E_{\rho})$, which follows from \eqref{eqn_rhobar_appendix}. This proves the desired result.
\end{proof}

\end{document}